\documentclass{article}

\usepackage[utf8]{inputenc}
\pdfoutput=1
\usepackage{amsmath}
\usepackage{amssymb}
\usepackage{amsthm}
\usepackage{microtype}
\usepackage{paralist}
\usepackage{subfigure}
\usepackage[sort,numbers]{natbib}
\usepackage[english]{babel}
\usepackage{tikz}
\usetikzlibrary{arrows,shapes,positioning,calc}
\usepackage[vlined,linesnumbered,ruled]{algorithm2e}
\usepackage{pgfplots}
\pgfplotsset{compat=1.4}
\usepgfplotslibrary{units}

\usepackage[pagebackref]{hyperref}

\hypersetup{
  colorlinks,
  linkcolor=red!30!black,
  citecolor=green!30!black, 
  urlcolor=blue!30!black,
  bookmarks=false,breaklinks,pdftitle={Towards Optimal and Expressive Kernelization for d-Hitting Set}, pdfauthor={Ren\'e van Bevern}, pdfdisplaydoctitle}

\makeatletter
\def\NAT@spacechar{~}%
\makeatother
\date{}

\newcommand{\decprob}[3]{%
  {\def\descriptionlabel##1{\hspace\labelsep\quad{}\it{}##1:}%
    \par\vspace{\topsep}\noindent
    \begin{minipage}{\textwidth}
      \quad \textsc{#1}%
      \begin{compactdesc}
      \item[Input] #2
      \item[Question] #3
      \end{compactdesc}
    \end{minipage}}\vspace{\topsep}}

\newcommand{\hs}[1]{$#1$\mbox{-}\textsc{Hit\-ting Set}}

\newcommand{\lt}{lin\-e\-ar-time}

\newtheorem{theorem}{Theorem}
\newtheorem{proposition}{Proposition}
\newtheorem{lemma}{Lemma}

\theoremstyle{definition}
\newtheorem{definition}{Definition}

\newtheorem{example}{Example}

\DeclareMathOperator\btrue{true}
\DeclareMathOperator\bfalse{false}
\DeclareMathOperator\bigTh{\Theta}
\DeclareMathOperator\bigOm{\Omega}
\DeclareMathOperator\bigO{O}

\DeclareMathOperator\petcount{petals}
\DeclareMathOperator\verts{vertices}
\DeclareMathOperator\intersect{intersection}

\DeclareMathOperator\markused{used}

\newcommand{\HG}{\ensuremath{H}}
\newcommand{\HGout}{\ensuremath{G}}
\newcommand{\BG}{\ensuremath{B}}
\newcommand{\Edgs}{\ensuremath{E}}
\newcommand{\Edg}{\ensuremath{e}}
\newcommand{\Verts}{\ensuremath{V}}
\newcommand{\PH}{the po\-ly\-nom\-i\-al-time hierarchy}

\newcommand{\vernum}{n} \newcommand{\edgenum}{m}

\newif\ifcolorprint \newif\ifarxiv \colorprintfalse \arxivtrue

\pgfdeclarelayer{background}
\pgfsetlayers{background,main}

\tikzstyle{vertex} = [fill,shape=circle,node distance=80pt]
\tikzstyle{edge} = [fill,opacity=.5,fill opacity=.5,line cap=round,
line join=round, line width=30pt] \tikzstyle{elabel} =
[fill,shape=circle,node distance=30pt]

\newcommand{\onepfive}{{1.5}} \newcommand{\cone}{{\ensuremath{b}}}
\newcommand{\ctwo}{{\ensuremath{c}}}

\newcommand{\wrel}{{\ensuremath{W}}} \newcommand{\uhs}{\textsc{Hitting
    Set}} 
\newcommand{\Petals}{\ensuremath{P}}
\newcommand{\Core}{\ensuremath{C}}

\newcommand{\dslice}{\ensuremath{\ell}}
\newcommand{\match}{\ensuremath{M}}
\newcommand{\matchvs}{\ensuremath{N}}

\title{Towards Optimal and Expressive Kernelization for $d$-Hitting
  Set\thanks{Supported by the DFG, project DAPA, NI~369/12. Parts of
    this work were done under DFG project AREG, NI~369/9. Earlier
    versions of this article appeared in the \emph{Proceedings of the
      18th Annual International Computing and Combinatorics Conference
      (COCOON'12)} and in the COCOON'12 special issue of
    Algorithmica~\citep{Bev13}.  Due to an implementation bug,
    Figure~3 of the Algorithmica article~\citep{Bev13} reports wrong
    sizes of the problem kernels obtained in the experiments and,
    consequently, the article gives a wrong interpretation of the
    experimental results.  This version gives corrected experimental
    results, adds additional figures, and more formally defines
    ``expressive kernelization''.}  }
\author{Ren\'e van Bevern\\Institut f\"ur Softwaretechnik und Theoretische Informatik\\
  TU Berlin, Germany, \\\texttt{rene.vanbevern@tu-berlin.de}}

\begin{document}

\def\algorithmautorefname{Al\-go\-rithm}%
\def\propositionautorefname{Prop\-o\-si\-tion}%
\def\constructionautorefname{Con\-struc\-tion}%
\def\definitionautorefname{Def\-i\-ni\-tion}%
\def\subfigureautorefname{Fig\-ure}%
\def\claimautorefname{Claim}%
\def\conjectureautorefname{Con\-jec\-ture}%
\def\lemmaautorefname{Lem\-ma}%
\def\exampleautorefname{Ex\-am\-ple}%
\def\corollaryautorefname{Cor\-ol\-lary}%
\def\observationautorefname{Ob\-ser\-va\-tion}%
\def\subsectionautorefname{Sec\-tion}%
\def\sectionautorefname{Sec\-tion}%
\def\chapterautorefname{Chap\-ter}%
\def\rruleautorefname{Re\-duc\-tion Rule}

\maketitle
\label{chap:hs}

\begin{abstract}
  \hs d is the NP-hard problem of selecting at most $k$~vertices of a
  hypergraph so that each hyperedge, all of which have cardinality at
  most~$d$, contains at least one selected vertex. The applications of
  \hs d are, for example, fault diagnosis, automatic program
  verification, and the noise-minimizing assignment of frequencies to
  radio transmitters.

  We show a linear-time algorithm that transforms an instance of \hs
  d %
  into an equivalent instance comprising at most
  $\bigO(k^d)$~hyperedges and vertices.  In terms of parameterized
  complexity, this is a \emph{problem
  kernel}. %
  Our kernelization algorithm is based on speeding up the well-known
  approach of finding and shrinking \emph{sunflowers} in
  hypergraphs, which yields problem kernels
  with structural properties that we condense into the concept of
  \emph{expressive kernelization}.

  We conduct experiments to show that our kernelization algorithm can
  kernelize instances with more than $10^7$~hyperedges in less than
  five minutes.

  Finally, we show that the number of vertices in the problem kernel
  can be further reduced to~$\bigO(k^{d-1})$ with additional
  $\bigO(k^{\onepfive d})$~processing time by nontrivially combining
  the sunflower technique with \hs d problem kernels due to Abu-Khzam
  and Moser.
\end{abstract}

\section{Introduction}
Many problems, like the examples given below, can be modeled as the
NP-hard \hs d problem: \decprob{\hs{d}}
{A hypergraph $\HG=(\Verts,\Edgs)$ with hyperedges whose cardinality
  is bounded from above by a constant~$d$, and a natural number~$k$.}
{Is there a \emph{hitting set} $S\subseteq\Verts$ with $|S|\leq k$ and
  $\forall \Edg\in\Edgs\colon\Edg\cap S\ne\emptyset$?}
Problems that can be modeled as \hs {d} arise,
among others, in the following fields.

\paragraph{Construction of Golomb rulers.} A \emph{Golomb ruler} of
length~$n$ is a subset of \emph{marks} $R\subseteq[n]$ such that no
pair of marks in~$R$ has the same distance as another pair.  The task
of finding shortest Golomb rulers with a fixed number of marks or
Golomb rulers of fixed length with a maximum number of marks arises,
among others, in radio frequency
allocation~\citep{Dra09}. \citet{SMNW14} showed how to construct
Golomb rulers using \hs4.  That is, $d=4$.

\paragraph{Fault diagnosis.} The task is to detect faulty components
of a malfunctioning system. To this end, those sets of components are
mapped to hyperedges of a hypergraph that are known to contain at
least one broken component~\citep{Rei87,dKW87}. By the principle of
Occam's razor, a small hitting set is then a likely explanation of the
malfunction.  In this application, $d$~is the maximum number of
components that a wrong observation depends on.

\paragraph{Program verification.} \citet{OC03} used \hs d in order to
automatically detect bugs in parallel Java programs while aiming for a
small slowdown of the program monitored at execution time.  In their
experiments, $d\leq 10$ was sufficient to debug complex software
suites.  In most cases, even $d\leq 5$ sufficed.  Remarkably, in this
application, one is interested in the question whether a hypergraph
allows for a hitting set of size at most~$k=d$, that is, both
$k$~and~$d$ are~small.

\bigskip\noindent The described problems have in common that a large
number of ``conflicts'' (the possibly $\bigO(n^d)$ hyperedges in a \hs
d instance) is caused by a small number of elements (the hitting
set~$S$), whose removal or repair could fix a broken system or
establish a useful property.

A powerful tool to attack NP-hard problems like \hs d is problem
kernelization~\citep{GN07,Kra14}---a form of provably efficient and
effective data reduction.  We show how to compute a problem kernel
with $\bigO(k^d)$~hyperedges for \hs d in linear time. We
experimentally evaluate our kernelization algorithm on \hs4 instances
arising in the construction of Golomb rulers with a maximum number of
marks and see that instances with more than~$10^7$ hyperedges are
kernelizable in less than five
minutes.  %

\paragraph{Known results.} \uhs{} is W[2]-complete with respect to the
parameter~$k$ when the cardinality of the hyperedges is
unbounded~\citep[Theorem~7.14]{FG06}. Hence, unless FPT${}={}$W[2], it
has no problem kernel. \citet{DvM14} showed that the existence of a
problem kernel with $\bigO(k^{d-\varepsilon})$~hyperedges for
any~$\varepsilon>0$ for \hs d implies a collapse of \PH{}. Therefore,
\hs d is assumed not to admit problem kernels with
$\bigO(k^{d-\varepsilon})$~hyperedges. For the same reason, \hs d
presumably has no polynomial-size problem kernels if $d$~is \emph{not}
constant.

Various problem kernels for \hs d have been
developed~\citep{NR03,NRT04,Dam06,FG06,Kra09,Abu10,Mos10,FK14}.
\citet{NR03} showed a problem kernel for \hs 3 of
size~$\bigO(k^3)$. They implicitly claimed that a polynomial-size
problem kernel for \hs d is computable in linear time, without giving
a proof for the running time. \citet{NRT04} claimed that a problem
kernel with $\bigO(k^{d-1})$~vertices is computable in
$\bigO(k(\vernum{}+\edgenum{})+k^d)$~time, which, however, does not
always yield correct problem kernels~\citep{Abu10}. %
\citet{Dam06} focused on developing small problem kernels for \hs d
and other problems with the focus on preserving \emph{all} minimal
solutions of size at most~$k$ (so-called \emph{full
  kernels}). \citet{FK14} presented a so-called \emph{streaming
  kernelization} for \hs d, which reads every hyperedge in the input
hypergraph at most once and has logarithmic memory usage for
fixed~$k$.  \citet{Abu10} showed a problem kernel
with~$\bigO(k^{d-1})$~vertices for \hs d, thus proving the previously
claimed result of \citet{NRT04} on the number of vertices in the
problem kernel. \citet[Section~7.3]{Mos10} built upon the work of
\citet{Abu10} to show a problem kernel for \hs d that also comprises
$\bigO(k^{d-1})$~vertices but, in contrast to the problem kernel of
\citet{Abu10}, yields a subgraph of the input hypergraph. The problem
kernels of \citet{Abu10} and \citet{Mos10} comprise
$\bigOm(k^{2d-2})$~hyperedges in the worst case.\footnote{Although not
  directly analyzed in the works of \citet{Abu10} and \citet{Mos10},
  this can be seen as follows: the kernel comprises vertices of a
  set~$W$ of ``weakly related'' hyperedges and an independent
  set~$I$. In the worst case, $|W|= k^{d-1}$, $|I|= dk^{d-1}$, and
  each hyperedge in~$W$ has $d$~subsets of size~$d-1$. Each such
  subset can constitute a hyperedge with each vertex in~$I$ and the
  kernel has $\bigOm(k^{2d-2})$~hyperedges.}

Several exponential-time algorithms for \textsc{Hitting Set} exist and
aim to decrease the exponential dependence of the running time on the
number of input vertices~\citep{SC10}, on the number of input
hyperedges~\citep{Fer06}, and on the size of the sought hitting
set~\citep{Fer10}. Also exponential-time approximation stepped into
the field of interest~\citep{BF11}, since, in polynomial time, \hs d
appears to be hard to approximate within a factor of better
than~$d$~\citep{KR08}.

\paragraph{Our results.} We show that a problem kernel for \hs d with
$\bigO(k^d)$~hyperedges and vertices is computable in linear time.
Thereby, we prove the previously claimed result by \citet{NR03} and
complement recent results in improving the efficiency of kernelization
algorithms~\citep{PSS09,vBHKNW11,Hag11,Kam13,FK14}.

Our problem kernel has useful structural properties that ensure the
interpretability of the problem kernel. We condense these
properties into the concept of \emph{expressive kernelization}.
Moreover, in the sense that a problem kernel with
$\bigO(k^{d-\varepsilon})$~hyperedges for some~$\varepsilon>0$ would
lead to a collapse of \PH{}, the size of our problem kernel is
optimal.

We implement our kernelization algorithm and evaluate its
applicability to the problem of constructing Golomb rulers with a
maximum number of marks and find instances with more than
$10^7$~hyperedges to be kernelizable in less than five~minutes.

Finally, using ideas of \citet{Abu10} and \citet{Mos10}, we show that
the number of vertices can be further reduced to $\bigO(k^{d-1})$ with
an additional amount of $\bigO(k^{\onepfive d})$~time. By merging
these techniques, we can compute in
$\bigO(\vernum{}+\edgenum{}+k^{\onepfive d})$~time a problem kernel
comprising $\bigO(k^d)$~hyperedges and $\bigO(k^{d-1})$~vertices.

\paragraph{Preliminaries.}\looseness=-1
A \emph{hypergraph}~$\HG=(\Verts,\Edgs)$ consists of a set of
\emph{vertices}~$\Verts$ and a set of \emph{(hyper)edges}~$\Edgs$,
where each hyperedge in $\Edgs$ is a subset of~$\Verts$.  We use
$n:=|\Verts|$ and $m:=|\Edgs|$. In a \emph{$d$-uniform hypergraph}
every edge has cardinality exactly~$d$. A 2-uniform hypergraph is a
\emph{graph}. A hypergraph~$\HGout=(\Verts',\Edgs')$ is a
\emph{subgraph} of its \emph{supergraph}~$\HG$ if
$\Verts'\subseteq\Verts$ and $\Edgs'\subseteq\Edgs$.  A
set~$S\subseteq \Verts$ intersecting every set in~$\Edgs$ is a
\emph{hitting set}.  A parameterized problem is a subset~$L\subseteq
\Sigma^*\times \mathbb N$~\citep{DF13,FG06,Nie06}.  A problem kernel
for a parameterized problem~$L$ is a polynomial-time algorithm that,
given an instance~$(I,k)$, computes an instance~$(I',k')$ such that
$|I'|+k'\leq f(k)$ and $(I', k')\in L\iff(I,k)\in L$. Herein, the
function~$f$ is called the \emph{size} of the problem kernel and
depends only on~$k$.

\paragraph{Paper outline.}
We start by giving a concept of expressive kernelization in
\autoref{sec:expressiveness}.

Then, we present an expressive linear-time kernelization algorithm for
\hs d in \autoref{sec:expl}, which we evaluate experimentally on
hypergraphs occurring in the computation of optimal Golomb rulers in
\autoref{sec:hsexp}.

Finally, we show how the number of vertices can be reduced
to~$\bigO(k^d)$ in additional $\bigO(k^{1.5k})$~time in
\autoref{sec:fewverts}.  Since the resulting problem kernel is not
expressive, we have not implemented it.

\section{Expressive kernelization}\label{sec:expressiveness}
The core component of our \lt{} kernelization algorithm for \hs d is
an algorithm to find and shrink \emph{sunflowers} in linear time.
Sunflowers are special constellations of hyperedges that \citet{ER60}
discovered to appear in any sufficiently large hypergraph and their
use in kernelization algorithms for \hs d is a standard
technique~\citep{FG06,Kra09,FSV13}.  They are defined as follows and
illustrated in \autoref{fig:bspsunflo}.

\begin{figure}\centering
  \begin{tikzpicture}[on grid, auto, node distance=1.5cm, inner
    sep=2pt]
    \tikzstyle{operator}=[circle,draw=black,fill=white,minimum
    size=6mm] \tikzstyle{edge} = [color=black,opacity=.2,line
    cap=round, line join=round, line width=22pt]
    
    \ifcolorprint \tikzstyle{s1style} = [color=magenta,opacity=.3]
    \tikzstyle{s2style} = [color=red,opacity=.3,line width=29pt]
    \tikzstyle{s3style} = [color=yellow,opacity=.4,line width=32pt]
    \tikzstyle{s4style} = [color=blue,opacity=.3, line width=35pt]
    \tikzstyle{s5style} = [color=green,opacity=.4,line width=26pt]
    \else \tikzstyle{s1style} = [] \tikzstyle{s2style} = [line
    width=29pt] \tikzstyle{s3style} = [line width=32pt]
    \tikzstyle{s4style} = [ line width=35pt] \tikzstyle{s5style} =
    [line width=26pt] \fi
    
    \node[operator] (w3) {}; \node[operator] (v3) [left of=w3] {};
    \node[operator] (v4) [below of=v3] {}; \node[operator] (w4) [left
    of=v3] {}; \node[operator] (w5) [left of=v4] {};
    
    \begin{pgfonlayer}{background}
      \draw[edge,s4style]
      (v3.center)--(w4.center)--(v4.center)--cycle;
      \draw[edge,s5style]
      (v3.center)--(v4.center)--(w5.center)--(v3.center);
      \draw[edge,s3style]
      (v4.center)--(w3.center)--(v3.center)--cycle;
    \end{pgfonlayer}
  \end{tikzpicture}
  \caption{A sunflower with three petals and two core elements.}
  \label{fig:bspsunflo}
\end{figure}
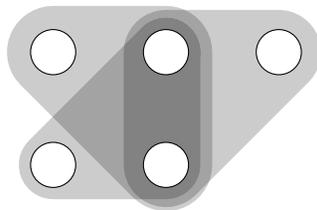

\begin{definition}
  A \emph{sunflower} in a hypergraph $\HG=(\Verts,\Edgs)$ is a set of
  \emph{petals}~$\Petals\subseteq \Edgs$ such that each pair of sets
  in~$\Petals$ intersects in exactly the same set~$\Core\subseteq
  \Verts$, which is called the \emph{core} (possibly,
  $C=\emptyset$). The \emph{size} of the sunflower is $|\Petals|$.
\end{definition}

\noindent The approach of finding and shrinking sunflowers yields
problem kernels that contain more structural information than the
formal definition of problem kernels requires.  Specifically,
sunflowers help computing problem kernels that have the following
three properties, which we henceforth require to be guaranteed by
\emph{expressive} problem kernels for \hs d and that we will describe
in more detail in the following.

\begin{definition}\label{exprdef}
  A kernelization algorithm for \hs d is \emph{expressive} if, given
  an instance~$(H,k)$, it outputs an instance~$(H',k')$ such that
  \begin{enumerate}[i)]
  \item\label{exprsub} $H'$ is a subgraph of~$H$,
  \item\label{exprmin} any vertex set of size at most~$k$ is a minimal
    hitting set for~$H$ if and only if it is a minimal hitting set
    for~$H'$, and
  \item\label{exprcert} it outputs a certificate for $(H',k')$~being a
    yes-instance if and only if $(H,k)$~is.
  \end{enumerate}
\end{definition}

\paragraph{Interpretability of the problem kernel.} The kernelization
algorithm should output a subgraph of the input hypergraph.
Kernelization algorithms for \hs d with this explicit goal have been
developed by \citet{Mos10} and \citet{Kra09}, since newly introduced
hyperedges or vertices in the problem kernel might not be
interpretable in the context of the original problem modeled as \hs
d. \citet{Kra09} exploited this property to show polynomial-size
problem kernels for a large class of problems formalizable as \hs d.
In some scenarios, as pointed out by \citet{AF06}, it is even
desirable that the kernelization algorithm outputs an \emph{induced}
subgraph of the input hypergraph.  However, our problem kernel for \hs
d will not satisfy this requirement.

\paragraph{Interpretability of solutions.} Any vertex set of size at
most~$k$ should be a minimal hitting set for the resulting problem
kernel if and only if it is a minimal hitting set for the original
instance.  If the input instance and the problem kernel allow for
exactly the same minimal hitting sets of size at most~$k$, the problem
kernel retains enough information for interpreting solutions and
finding alternative solutions without having to consider the input
hypergraph.  This property has been exploited by \citet{FSV13} as an
important building block in a polynomial-size problem kernel for a
problem that cannot easily be modeled as \hs d for
\emph{constant}~$d$.  As pointed out by \citet{FSV13}, this property
is stronger than those guaranteed by the \emph{full kernels}
introduced by \citet{Dam06}: full kernels contain all minimal hitting
sets of size at most~$k$ for the input hypergraph, but not necessarily
the information whether a hitting set is minimal.

\paragraph{Certifying data reduction.}  Similarly to how certifying
algorithms provide a certificate for the correctness of their
output~\citep{MMNS11}, an expressive kernelization algorithm should
provide a certificate for the correctness of the executed data
reduction.  Ideally, the proof that a certificate indeed certifies the
correctness should be easily understandable, so that a human can
easily verify the executed data reduction to be correct without having
to trust on the correctness of algorithms and their implementations.
A~sunflower~$\Petals$ with $k+1$~petals in a \hs d instance fulfills
this requirement: every hitting set~$S$ of size at most~$k$ contains
an element of the core~$C$ of~$\Petals$, since otherwise $S$~cannot
contain an element of each of the $k+1$~petals.  Thus, any additional
hyperedge in the hypergraph that contains~$C$ already contains an
element of~$S$; it is redundant and may be removed. The
sunflower~$\Petals$ is a certificate
for this being correct.

\begin{figure}
  \subfigure[A Boolean circuit with circle nodes representing gates
  and square nodes representing input and output nodes.]{
    \label{fig:circuit}
    \begin{tikzpicture}[->, shorten >=2pt, node distance=1.25cm,inner
      sep=2pt]
      \tikzstyle{inval}=[rectangle,draw=black,fill=gray!20,minimum
      size=6mm]
      \tikzstyle{outval}=[rectangle,draw=black,fill=gray!20,minimum
      size=6mm]
      \tikzstyle{operator}=[circle,draw=black,fill=gray!20,minimum
      size=6mm]

      \begin{scope}[yshift=0cm]
        \node[inval] (x2) {$x_1$}; \node[operator] (v2) [right
        of=x2]{$v_1$}; \node[operator] (w2) [right of=v2]{$w_1$};
        \node[outval] (fx2) [right of=w2]{$f_1(x)$};
      \end{scope}

      \begin{scope}[yshift=-0.8cm]
        \node[inval] (x4) {$x_2$}; \node[operator] (v4) [right
        of=x4]{$v_2$}; \node[operator] (w5) [right of=v4]{$w_2$};
        \node[outval] (fx5) [right of=w5]{$f_2(x)$};
      \end{scope}

      \begin{scope}[yshift=-1.6cm]
        \node (x5) {}; \node (v5) [right of=x5]{}; \node[operator]
        (w4) [right of=v5]{$w_3$}; \node[outval] (fx4) [right
        of=w4]{$f_3(x)$};
      \end{scope}

      \begin{scope}[yshift=-2.4cm]
        \node[inval] (x3) {$x_3$}; \node[operator] (v3) [right
        of=x3]{$v_3$}; \node[operator] (w3) [right of=v3]{$w_4$};
        \node[outval] (fx3) [right of=w3]{$f_4(x)$};
      \end{scope}

      \begin{scope}[yshift=-3.2cm]
        \node[inval] (x1) {$x_4$}; \node[operator] (v1) [right
        of=x1]{$v_4$}; \node[operator] (w1) [right of=v1]{$w_5$};
        \node[outval] (fx1) [right of=w1]{$f_5(x)$};
      \end{scope}

      \foreach \to in {v1,v3} { \draw (x1) -> (\to); }

      \foreach \to in {v2,v4} { \draw (x2) -> (\to);}

      \foreach \to in {v3,v1} { \draw (x3) -> (\to); }

      \foreach \to in {v4,v2} { \draw (x4) -> (\to); }

      \foreach \to in {w3,w1,w4} { \draw (v1) -> (\to); }

      \foreach \to in {w1,w3,w4,w5} { \draw (v3) -> (\to); }
      \foreach \to in {w2} { \draw (v2) -> (\to); }

      \foreach \to in {w2,w5} { \draw (v4) -> (\to); }

      \foreach \from/\to in {w1/fx1,w2/fx2,w3/fx3,w4/fx4,w5/fx5} {
        \draw (\from) -> (\to); }
    \end{tikzpicture}
  }\hfill \subfigure[Sets containing at least one faulty gate, found
  by the analysis of the circuit.]{
    \label{fig:hg}
    \begin{tikzpicture}[on grid, auto, node distance=1.5cm, inner
      sep=2pt]
      \tikzstyle{operator}=[circle,draw=black,fill=white,minimum
      size=6mm] \tikzstyle{edge} = [color=black,opacity=.2,line
      cap=round, line join=round, line width=22pt]

      \ifcolorprint \tikzstyle{s1style} = [color=magenta,opacity=.3]
      \tikzstyle{s2style} = [color=red,opacity=.3,line width=29pt]
      \tikzstyle{s3style} = [color=yellow,opacity=.4,line width=32pt]
      \tikzstyle{s4style} = [color=blue,opacity=.3]
      \tikzstyle{s5style} = [color=green,opacity=.4,line width=26pt]
      \else \tikzstyle{s1style} = [] \tikzstyle{s2style} = [line
      width=29pt] \tikzstyle{s3style} = [line width=32pt]
      \tikzstyle{s4style} = [] \tikzstyle{s5style} = [line width=26pt]
      \fi

      \node[operator] (w3) { $w_3$ }; \node[operator] (v3) [left
      of=w3] { $v_3$ }; \node[operator] (v4) [below of=v3] { $v_4$ };
      \node[operator] (w4) [below of=v4] { $w_4$ }; \node[operator]
      (w5) [left of=v4] { $w_5$ }; \node[operator] (w2) [left of=v3] {
        $w_2$ }; \node[operator] (v2) [left of=w2] { $v_2$ };
      \node[operator] (v1) [below of=v2] { $v_1$ }; \node[operator]
      (w1) [below of=v1] { $w_1$ };

      \node(S1) at ($(w1)!.5!(v1)$) {$S_1$}; \node(S2) at
      ($(v2)!.525!(w2)$) {$S_2$}; \node(S3) at ($(w3)!.45!(v3)$)
      {$S_3$}; \node(S4) at ($(w4)!.45!(v4)$) {$S_4$}; \node(S5) at
      ($(w5)!.45!(v4)$) {$S_5$};
      \begin{pgfonlayer}{background}
        \draw[edge,s1style] (v2.center)--(w1.center);
        \draw[edge,s2style] (v2.center)--(v3.center);
        \draw[edge,s4style] (v3.center)--(w4.center);
        \draw[edge,s5style]
        (v3.center)--(v4.center)--(w5.center)--(v3.center);
        \draw[edge,s3style]
        (v4.center)--(w3.center)--(v3.center)--cycle;
      \end{pgfonlayer}

    \end{tikzpicture}
  }
  \caption{Illustrations for \autoref{ex:ex2}.}
\end{figure}

\begin{example}\label{ex:ex2}
  Sunflowers not only certify the correctness of data reduction, but
  also lead the way to alternative solutions.  We illustrate this
  using an example of \hs d in a fault diagnosis context.

  \autoref{fig:circuit} represents a Boolean circuit.  It gets as
  input a 4-bit string~$x=x_1\dots x_4$ and outputs a 5-bit
  string~$f(x)=f_1(x)\dots f_5(x)$. The nodes drawn as circles
  represent Boolean gates, which output some bit depending on their
  two input bits. They might, for example, represent the logical
  operators~$\wedge$ or~$\vee$.  Assume that all output bits of~$f(x)$
  are observed to be the opposite of what would have been expected by
  the designer of the circuit. We want to identify broken gates. For
  each wrong output bit~$f_i(x)$, we obtain a set~$S_i$ of gates for
  which we know that at least one is broken because $f_i(x)$~is
  wrong. That is, $S_i$~contains precisely those gates that have a
  directed path to~$f_i(x)$ in the graph shown in
  \autoref{fig:circuit}. We obtain the sets~$S_1,\dots,S_5$
  illustrated in \autoref{fig:hg}.

  The sets~$S_1$ and~$S_4$ are disjoint. Therefore, the wrong output
  is not explainable by only one broken gate. Hence, we assume that
  there are \emph{two} broken gates and search for a hitting set of
  size~$k=2$ in the hypergraph with the vertices
  $v_1,\dots,v_4,w_1,\dots,w_5$ and hyperedges~$S_1,\dots,S_5$. The
  set~$\{S_3,S_4,S_5\}$ is a sunflower of size~$k+1=3$ with
  core~$\{v_3,v_4\}$.  Therefore, the functionality of gate~$v_3$
  and~$v_4$ must be checked.  If, in contrast to our expectations,
  both gates~$v_3$ and~$v_4$ turn out to be working correctly, the
  sunflower shows not only that at least three gates are broken, but
  also shows which gates have to be checked for malfunctions next:
  $w_3$, $w_4$, and~$w_5$.
\end{example}
\noindent There are few expressive kernelization algorithms for \hs d
in the literature.  For example, the algorithm of \citet{Abu10} does
not yield a subgraph of the input hypergraph.  Thus, it does not
satisfy \autoref{exprdef}\eqref{exprsub}, which has been ``fixed'' by
\citet{Mos10}. However, both problem kernels may discard minimal
solutions of size at most~$k$ and, thus, do not satisfy
\autoref{exprdef}\eqref{exprmin}.  \citet{Dam06} designed a problem
kernel that retains all minimal solutions of size at most~$k$.
However, it is not certifying and thus does not satisfy
\autoref{exprdef}\eqref{exprcert}.

\looseness=-1 We are aware of only one expressive kernelization algorithm for \hs d:
this is the problem kernel shown by \citet{Kra09}, which is also used
by \citet{FSV13}.  This is precisely the algorithm we will improve to
run in linear~time.
\section{A linear-time kernelization algorithm}
\label{sec:expl}
This section shows a linear-time computable problem kernel for \hs d
comprising $\bigO(k^d)$~hyperedges. That is, we show that a
hypergraph~$\HG$ can be transformed in linear time into a
hypergraph~$\HGout{}$ such that $\HGout{}$ has $\bigO(k^d)$~hyperedges
and allows for a hitting set of size~$k$ if and only $\HG$~does. In
\autoref{sec:fewverts}, we show how to shrink the number of vertices
to~$\bigO(k^{d-1})$.
\begin{theorem}\label{thm:optkern-lintime}
  \hs d allows for an expressive problem kernel with $d!\cdot d^{d+1}
  \cdot (k+1)^d$~hyperedges and $d$ times as many vertices that is
  computable in $\bigO(d\cdot \vernum{}+2^dd\cdot
  \edgenum{})$~time. %
\end{theorem}

\noindent We prove \autoref{thm:optkern-lintime} with the help of the
sunflower lemma of \citet{ER60}, who showed that every sufficiently
large hypergraph contains a sunflower with $k+2$~petals: if we shrink
all of these sunflowers, it follows that the resulting hypergraph will
be small. Kernelization algorithms based on this strategy, like those
of \citet{FG06} and \citet{Kra09} usually proceed along the lines of
repeatedly
\begin{itemize}
\item finding a sunflower of size~$k+2$ in the input hypergraph and
\item deleting redundant petals until no more sunflowers of size~$k+2$
  exist.
\end{itemize}
This approach has the drawback of finding only one sunflower at a time
and restarting the process from the beginning.

In contrast, to prove \autoref{thm:optkern-lintime}, we construct a
subgraph~$\HGout{}$ of a given hypergraph~$\HG$ not by hyperedge
deletion, but by a bottom-up approach that allows us to ``grow'' many
sunflowers in~$\HGout$ simultaneously, stopping ``growing sunflowers''
when they become too large.

\autoref{alg:eff-sunflo} repeatedly (after some initialization work in
lines \ref{lin:initstart}--\ref{lin:initend}) in \autoref{lin:addH}
copies a hyperedge~$\Edg$ from~$\HG$ to the initially empty~$\HGout{}$
unless we find in \autoref{lin:necess} that $\Edg$~contains the
core~$\Core$ of a sunflower of size~$k+1$ in~$\HGout{}$.  We maintain
the number of petals found for a core~$\Core$ in~$\petcount[\Core]$.
If we find that a hyperedge~$\Edg$ can be added to a sunflower with
core~$\Core$ in \autoref{lin:Hsuitable}, then we
increment~$\petcount[\Core]$ in \autoref{lin:incC} and mark the
vertices in $\Edg\setminus \Core$ as ``used'' for the core~$\Core$ in
\autoref{lin:markused}. This information is maintained by setting
``$\markused[\Core][v]\gets\btrue$.'' In this way, vertices
in~$\Edg\setminus \Core$ are not considered again for finding petals
for the core~$\Core$ in \autoref{lin:Hsuitable}, therefore ensuring
that petals later found for the core~$\Core$ intersect~$\Edg$ only
in~$\Core$.

\begin{algorithm}[h]\small
  \caption{Linear-Time kernelization for \hs d}
  \label{alg:eff-sunflo}
  \KwIn{A hypergraph~$\HG=(\Verts,\Edgs)$ and a natural number~$k$.}
  \KwOut{A hypergraph~$\HGout{}=(\Verts',\Edgs')$ with $|\Edgs'|\in
    \bigO(k^d)$.}  $\Edgs'\gets\emptyset$\nllabel{lin:initstart}\;
  \ForEach(\tcp*[f]{\textrm{Initialization for each
      hyperedge}}){$\Edg\in\Edgs$}{
    \ForEach(\tcp*[f]{\textrm{Initialization for all possible cores of
        sunflowers}}){$\Core\subseteq\Edg$}{ $\petcount[\Core]\gets
      0$\tcp*{\textrm{No petals found for sunflower with core~$\Core$
          yet}} \ForEach(\nllabel{lin:setfalse}){$v\in\Edg$}{
        $\markused[\Core][v]\gets\bfalse$\tcp*{\textrm{No vertex~$v$
            is in a petal of a\\sunflower with core~$\Core$ yet}}} }
    \nllabel{lin:initend} }
  \ForEach(\nllabel{lin:eachedge}){$\Edg\in\Edgs$}{
    \If{$\forall\Core\subseteq \Edg\colon\petcount[\Core]\leq
      k$\nllabel{lin:necess}}{
      $\Edgs'\gets\Edgs'\cup\{\Edg\}$\nllabel{lin:addH}\;
      \ForEach(\tcp*[f]{\textrm{Consider all possible cores for the
          petal~$\Edg$}}\nllabel{lin:eachcore}){$\Core\subseteq\Edg$}{
        \If{$\forall v\in \Edg\setminus
          \Core\colon\markused[\Core][v]=\bfalse$\nllabel{lin:Hsuitable}}{
          $\petcount[\Core]\gets\petcount[\Core]+1$\nllabel{lin:incC}\;
          \lForEach(\nllabel{lin:markused}){$v\in\Edg\setminus
            \Core$}{$\markused[\Core][v]\gets\btrue$} } } } }
  $\Verts':=\bigcup_{\Edg\in\Edgs'}\Edg$\nllabel{lin:vertunion}\;
  \Return{$(\Verts',\Edgs')$}\;
\end{algorithm}

\noindent By storing in $\petcount[\Core]$ a list of found petals, the
algorithm can output the discovered sunflowers without any increase in
running time.  Thus, it gives certificates for the correctness of the
executed data reduction.

It is important to note that, as illustrated in
\autoref{fig:algoillu}, the value in~$\petcount[\Core]$ is not
necessarily the size of the largest possible sunflower with
core~$\Core$, but depends on the order in that
\autoref{alg:eff-sunflo} processes the hyperedges of the input
hypergraph.  Computing the size of a largest sunflower with
core~$\Core$ is, for~$\Core=\emptyset$, the problem of computing a
maximum matching in a hypergraph, which is NP-hard \citep{Kar72}.

\begin{figure}
  \centering
  \begin{tikzpicture}[x=1cm, y=1cm]
    \tikzstyle{operator}=[circle,draw=black,fill=white,minimum
    size=6mm] \tikzstyle{edge} = [color=black,opacity=.2,line
    cap=round, line join=round, line width=22pt]

    \node[operator] (v) at (0,0) {$v$}; \node at (90:0.9) {$e_3$};
    \node[operator] (up1) at (90:1.5) {}; \node[operator] (up2) at
    (90:2.5) {$w_3$}; \node at (0:0.9) {$e_2$}; \node[operator] (r1)
    at (0:1.5) {}; \node[operator] (r2) at (0:2.5) {$w_2$}; \node at
    (270:0.9) {$e_1$}; \node[operator] (d1) at (270:1.5) {};
    \node[operator] (d2) at (270:2.5) {$w_1$}; \node[operator] (l2) at
    (180:2.5) {}; \node at (180:1.25) {$e_4$};

    \begin{pgfonlayer}{background}
      \draw[edge, line width=34] (v.center)--(up2.center); \draw[edge,
      line width=30] (v.center)--(d2.center); \draw[edge, line
      width=26] (v.center)--(r2.center);

      \draw[edge, line width=38] (v.center)--(l2.center)
      to[out=90,in=180](up2.center) to[out=0,in=90](r2.center)
      to[out=270,in=0](d2.center);
    \end{pgfonlayer}
  \end{tikzpicture}
  \caption{The result of applying \autoref{alg:eff-sunflo} depends on
    the order in that it processes the hyperedges of the input
    hypergraph~$\HG{}$.  If applied for~$k=2$, then
    \autoref{alg:eff-sunflo} will not add the hyperedge~$e_4$ to the
    output hypergraph~$\HGout{}$ if it before added~$e_1$, $e_2$,
    and~$e_3$, since it discovers~$e_4$ to contain the core~$\{v\}$ of
    the sunflower~$\{e_1,e_2,e_3\}$ with~$k+1=3$ petals.  However, if
    it first adds~$e_4$ to~$\HGout{}$, then it marks the
    vertices~$w_1$, $w_2$, and~$w_3$ as used for the sunflower with
    core~$\{v\}$. Thus, none of the hyperedges~$e_1$, $e_2$, and~$e_3$
    is recognized as a petal for a sunflower with core~$\{v\}$ and all
    shown hyperedges are added to the output hypergraph.}
  \label{fig:algoillu}
\end{figure}
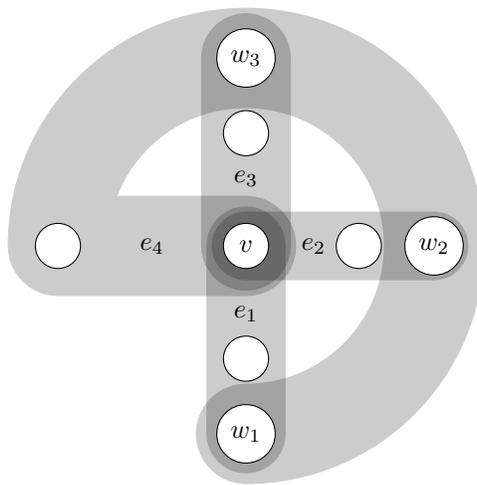

Towards proving \autoref{thm:optkern-lintime}, we now proceed as
follows. \autoref{sec:kern-correct} shows that
\autoref{alg:eff-sunflo} is correct and expressive.
\autoref{sec:hs-iskern} shows that the hypergraph output by
\autoref{alg:eff-sunflo} contains $\bigO(k^d)$~hyperedges. Finally,
\autoref{sec:algfast} shows that \autoref{alg:eff-sunflo} runs in
linear time.
\subsection{Correctness}\label{sec:kern-correct}

On our way to proving that \hs d has an expressive linear-time
computable problem kernel with $\bigO(k^d)$~hyperedges and thus
proving \autoref{thm:optkern-lintime}, we now prove the correctness
and expressiveness of \autoref{alg:eff-sunflo}. This, together with a
proof for the size of the output hypergraph and a proof of the running
time of \autoref{alg:eff-sunflo}, will provide a proof of
\autoref{thm:optkern-lintime}.

\begin{proposition}\label{lem:kern-correct}\label{prop:corr}
  Let $\HGout{}$~be the hypergraph returned by
  \autoref{alg:eff-sunflo} when given a hypergraph~$\HG$ and an
  integer~$k$. Then,
  \begin{enumerate}[i)]
  \item\label{minhs1} any hitting set~$S$ of size at most~$k$
    for~$\HGout{}$ is a hitting set for~$\HG$ and
  \item\label{minhs2} any minimal hitting set~$S$ of size at most~$k$
    for~$\HG$ is a hitting set for~$\HGout{}$.
  \end{enumerate}
  Moreover, $\HGout{}$~is a subgraph of~$\HG$ and any subset~$S$ of at
  most $k$~vertices of~$\HG$ is a minimal hitting set for~$\HG$ if and
  only if it is a minimal hitting set for~$\HGout{}$.
\end{proposition}

\begin{proof}
  By construction of~$\HGout{}$ from~$\HG{}$ in
  \autoref{alg:eff-sunflo}, it is clear that $\HGout{}$~is a subgraph
  of~$\HG{}$.  We now show that it is sufficient to prove
  \eqref{minhs1} and \eqref{minhs2} to also conclude the last
  statement of \autoref{prop:corr}: let $S$~be a minimal hitting set
  of size at most~$k$ for~$\HG$. By (\ref{minhs2}), it is a hitting
  set for~$\HGout{}$. Assume, towards a contradiction, that $S$~is not
  a \emph{minimal} hitting set for~$\HGout${}. Then, there is a
  hitting set~$S'\subsetneq S$ for~$\HGout$.  However, by
  (\ref{minhs1}), $S'\subsetneq S$~is also a hitting set
  for~$\HG$. This contradicts $S$~being a minimal hitting set
  for~$\HG$.  Symmetrically, let $S$~be a minimal hitting set~$S$ of
  size at most~$k$ for $\HGout{}$. By (\ref{minhs1}), $S$~is a hitting
  set for~$\HG{}$.  Assume, towards a contradiction, that $S$~is not a
  \emph{minimal} hitting set for~$\HG$. Then, there is a minimal
  hitting set~$S'\subsetneq S$ for~$\HG$. However, by~(\ref{minhs2}),
  $S'\subsetneq S$~is also a hitting set for~$\HGout{}$. This
  contradicts $S$~being a minimal hitting set for~$\HGout{}$.  It
  remains to prove (\ref{minhs1}) and (\ref{minhs2}).

  (\ref{minhs1}) Let $S$~be a hitting set of size at most~$k$
  for~$\HGout{}$. Obviously, all hyperedges that~$\HG$ and~$\HGout{}$
  have in common are hit in~$\HG$ by~$S$. We show that every
  hyperedge~$\Edg$ in~$\HG$ that is not in~$\HGout{}$ is also hit.  If
  $\Edg$ is in~$\HG{}$ but not in~$\HGout{}$, then adding $\Edg$
  to~$\HGout{}$ in \autoref{lin:addH} of \autoref{alg:eff-sunflo} has
  been skipped because the condition in \autoref{lin:necess} is
  false. That is, $\petcount[\Core]\geq k+1$ for some~$\Core\subseteq
  \Edg$. Consequently, for this particular~$\Core$, a
  sunflower~$\Petals$ with $k+1$ petals and core~$\Core$ exists
  in~$\HGout{}$, since we only increment $\petcount[\Core]$ in
  \autoref{lin:incC} if we find a suitable additional petal for the
  sunflower with core~$\Core$ in \autoref{lin:Hsuitable}.  Note that
  $\Core\ne\emptyset$ because, otherwise, $k+1$~pairwise disjoint
  hyperedges would exist in~$\HGout{}$, contradicting our assumption
  that~$S$ is a hitting set of size~$k$ for~$\HGout{}$. Since $|S|\leq
  k$, we have $S\cap \Core\ne\emptyset$. Therefore, since
  $\Core\subseteq\Edg$, the hyperedge~$\Edg$ is hit by~$S$ also
  in~$\HG$.

  (\ref{minhs2}) Let $S$~be a minimal hitting set of size at most~$k$
  for~$\HG=(\Verts,\Edgs)$. The set~$S':=S\cap\Verts'$ is a hitting
  set for~$\HGout{}=(\Verts',\Edgs')$ with $S'\subseteq S$: the
  set~$S$ contains an element of every hyperedge in~$\Edgs$ and, since
  $\Edgs'\subseteq\Edgs$ and $V'=\bigcup_{\Edg\in\Edgs'}\Edg$, the
  set~$S'$ contains an element of every hyperedge in~$\Edgs'$.  By
  (\ref{minhs1}), $S'$~is also a hitting set for~$\HG$. Since
  $S'\subseteq S$ and we required $S$~to be a \emph{minimal} hitting
  set of size at most~$k$ for~$\HG$, we have that $S'=S$ and, thus,
  that $S$~is a hitting set for~$\HGout{}$.
\end{proof}

\subsection{Problem kernel size}\label{sec:hs-iskern}

Having shown that \autoref{alg:eff-sunflo} is correct, we now show
that the hypergraph output by \autoref{alg:eff-sunflo} contains
$\bigO(k^d)$~hyperedges.  To prove \autoref{thm:optkern-lintime}, it
then remains to prove that \autoref{alg:eff-sunflo} runs in linear
time.

In order to show an upper bound on the size of the hypergraph output
by \autoref{alg:eff-sunflo}, we exploit an upper bound on the size of
the sunflowers in the output hypergraph:

\begin{lemma}\label{lem:few-petals}
  Let $\HGout$ be the hypergraph output by \autoref{alg:eff-sunflo}
  applied to a hypergraph~$\HG$ and a natural number~$k$. Every
  sunflower~$\Petals$ in~$\HGout$ with core~$\Core\notin\Petals$ has
  size at most~$d(k+1)$.
\end{lemma}

\begin{proof}
  Let $\Petals$ be a sunflower in~$\HGout$ with
  core~$\Core\notin\Petals$. Then, $|\Petals|\leq d(k+1)$ follows from
  the following two observations:

  (i) Every petal~$\Edg\in\Petals$ present in~$\HGout{}$ is copied
  from~$\HG$ in \autoref{lin:addH} of
  \autoref{alg:eff-sunflo}. Consequently, every petal~$\Edg\in\Petals$
  contains a vertex~$v$ satisfying $\markused[\Core][v]=\btrue$: if
  this condition is violated in \autoref{lin:Hsuitable}, then
  \autoref{lin:markused} applies ``$\markused[\Core][v]\gets\btrue$''
  to all vertices~$v\in\Edg\setminus\Core$.

  (ii) Whenever $\petcount[\Core]$ is incremented by one in
  \autoref{lin:incC}, then, in \autoref{lin:markused},
  ``$\markused[\Core][v]\gets\btrue$'' is applied to the at most
  $d$~vertices~$v\in\Edg$. Thus, since $\petcount[\Core]$ never
  exceeds~$k+1$, at most $d(k+1)$~vertices~$v$ satisfy
  $\markused[\Core][v]=\btrue$. Moreover, since, by
  \autoref{lin:markused}, no~$v\in C$ satisfies
  $\markused[\Core][v]=\btrue$ and the petals in~$\Petals$ pairwise
  intersect only in~$\Core$, it follows that at most $d(k+1)$~petals
  in~$\Petals$ contain vertices satisfying
  $\markused[\Core][v]=\btrue$.
\end{proof}
\noindent Having shown an upper bound on the size of the sunflowers in
the hypergraph output by \autoref{alg:eff-sunflo}, we now show that
the output hypergraph contains $\bigO(k^d)$~hyperedges.  To this end,
in a way similar to \citet[Lemma~9.7]{FG06}, we show the following
refined version of \citet{ER60}'s sunflower lemma.  Herein, recall
that a hypergraph is $\dslice$-uniform if and only if every hyperedge
has cardinality exactly~$\dslice$.

\begin{lemma}\label{lem:largeflowers}
  Let $\HG$~be an $\dslice$-uniform hypergraph and
  $\cone,\ctwo\in\mathbb N$ with $\cone\leq\dslice$ such that every
  pair of hyperedges in~$\HG$ intersects in at most
  $\dslice-\cone$~vertices.

  If $\HG$~contains more than
  $\dslice!\ctwo^{\dslice+1-\cone}$~hyperedges, then $\HG$ contains a
  sunflower with more than $\ctwo$ petals.
\end{lemma}

\noindent For $\cone=1$, we obtain the sunflower lemma stated by
\citet{FG06}.  For $\cone=2$, we will exploit it in
\autoref{sec:fewverts} to reduce the number of vertices in the output
hypergraph.

\begin{proof}
  We prove the lemma by induction on~$\ell$. As base case, consider
  $\dslice=\cone$. For~$\dslice=\cone$, all hyperedges in~$\HG$ are
  pairwise disjoint. Hence, if $\HG$~has more than
  $\ell!c^{\ell+1-b}$~hyperedges, then these form a sunflower with
  empty core and more than $\ell!c^{\ell+1-b}=\ell!c\geq
  c$~petals. That is, the lemma holds for $\dslice=\cone$.

  \looseness=-1 Now, assume that the lemma holds for some~$\dslice\geq\cone$. It
  remains to prove that it holds for~${\dslice+1}$. Let $\match$~be a
  maximal set of pairwise disjoint hyperedges of the
  $(\dslice+1)$-uniform hypergraph~$\HG:=(\Verts,\Edgs)$. If
  $|\match|>\ctwo$, then the lemma holds because~$M$ is a sunflower
  with empty core. Otherwise, for $\matchvs:=\bigcup_{\Edg\in \match}
  \Edg$, it holds that $|\matchvs|\leq(\dslice+1)\ctwo$ and some
  vertex $w\in \matchvs$ is contained in a set~$\Edgs_w$ of more than
  \begin{align*}
    \frac{|\Edgs|}{|\matchvs|} &\geq
    \frac{(\ell+1)!\ctwo^{\ell+2-\cone}} {(\dslice+1)\ctwo} =
    \ell!\ctwo^{\ell+1-\cone}\text{ hyperedges.}
  \end{align*}
  The hypergraph $\HG_w$ that contains for each
  hyperedge~$\Edg\in\Edgs_w$ the hyperedge~$\Edg\setminus\{w\}$ is an
  $\dslice$-uniform hypergraph and, by induction hypothesis, contains
  a sunflower~$\Petals$ with more than $\ctwo$~petals. Adding $w$ to
  each of the petals of~$\Petals$ yields a sunflower~$\Petals'$ with
  more than $\ctwo$~petals in~$\HG$.
\end{proof}

\noindent By combining \autoref{lem:few-petals} with
\autoref{lem:largeflowers}, we can easily show that the hypergraph
output by \autoref{alg:eff-sunflo} contains $\bigO(k^d)$~hyperedges.
Since we have already shown in \autoref{prop:corr} that the algorithm
is correct, it thereafter only remains to show that
\autoref{alg:eff-sunflo} runs in linear time in order to complete the
proof of \autoref{thm:optkern-lintime}.

\begin{proposition}\label{prop:small}
  The hypergraph~$\HGout{}$ returned by \autoref{alg:eff-sunflo} on
  input $\HG$ and~$k$ contains at most $d!\cdot d^{d+1} \cdot
  (k+1)^d$~hyperedges and $d$ times as many vertices.
\end{proposition}

\begin{proof}
  Obviously, $\HGout{}$ has at most $d$~times as many vertices as
  hyperedges, since the vertex set of~$\HGout{}$ is constructed as the
  union of its hyperedges in \autoref{lin:vertunion}
  of~\autoref{alg:eff-sunflo}.

  To bound the number of hyperedges, consider, for~$1\leq\dslice\leq
  d$, the $\dslice$-uniform
  hypergraph~$\HGout_{\dslice}=(\Verts_\dslice,\Edgs_\dslice)$
  comprising only the hyperedges of size~$\dslice$
  of~$\HGout{}$. If~$\HGout{}$ had more than $d!\cdot d^{d+1} \cdot
  (k+1)^d$ hyperedges, then, for some $\dslice\leq d$,
  $\HGout{}_{\dslice}$~would have more than $d!\cdot d^d \cdot
  (k+1)^d$ hyperedges.  \autoref{lem:largeflowers} with $\cone=1$
  and~$\ctwo=d(k+1)$ states that, if $\HGout{}_{\dslice}$ had more
  than $\dslice!\cdot d^\dslice \cdot (k+1)^\dslice$ hyperedges, then
  $\HGout{}_{\dslice}$ would contain a sunflower~$\Petals$ with
  core~$\Core$ and more than $d(k+1)$ petals.  Obviously,
  $\Core\notin\Petals$, since all petals have
  cardinality~$\dslice$. Moreover, this sunflower would also exist in
  the supergraph~$\HGout{}$ of~$\HGout{}_{\dslice}$, contradicting
  \autoref{lem:few-petals}.
\end{proof}

\subsection{Running time}\label{sec:algfast}

\looseness=-1 Since \autoref{prop:corr} has shown that
\autoref{alg:eff-sunflo} is correct and \autoref{prop:small} has shown
that the output hypergraph contains $\bigO(k^d)$~hyperedges, to prove
\autoref{thm:optkern-lintime}, it remains to show that
\autoref{alg:eff-sunflo} runs in linear time.  In order to implement
the algorithm efficiently, we need data structures that allow us to
quickly look up the values~$\petcount[\Core]$ and~$\markused[\Core]$
for some vertex set~$\Core$ of size~at~most~$d$.

The usual approach to realize table look-ups for subsets of some universe
of size~$\gamma$ in $\bigO(\gamma)$~time is representing the subsets
as bitstrings of length~$\gamma$ and looking up these in a trie
\citep[Section~5.3]{AHU83}. However, here, our universe is the set of
vertices of the input hypergraph and, thus, has size~$n$. Hence, this
method would yield table look-ups in $\bigTh(n)$~time, which is too
slow to prove that \autoref{alg:eff-sunflo} runs in linear time.  For
this reason, we will not represent vertex subsets of size at most~$d$
as bitstrings, but uniquely represent them as sorted sequences of at
most $d$~integers. Then, we will exploit the following lemma.

\begin{lemma}\label{lem:stringtrie}
  Let $L$~be a list of sequences of length at most~$d$ of integers
  in~$[n]$.

  \nopagebreak In $\bigO(d\cdot n+d\cdot |L|)$~time, we can compute an
  associative array~$A[]$ such that, for each sequence~$s$ in~$L$,
  accessing the value~$A[s]$ and storing a value to~$A[s]$ works in
  $\bigO(d)$~time.
\end{lemma}

\begin{proof}
  We use a trie to associate values with sequences in~$L$.  However,
  the trie will be too large to initialize it fully in linear time.
  We have to show that we can create the trie so that, for a look-up
  of a value for any sequence~$s$ in~$L$, no uninitialized memory
  cells are read.
  \begin{figure}
    \centering \tikzset{ trienode/.style={%
        draw, rectangle, rectangle split, rectangle split parts=5,
        rectangle split part align=top,
        rectangle split ignore empty parts=false, font=\scriptsize,
        inner sep=2pt,
        every one node part/.style={text width=0.5cm}, every two node
        part/.style={text width=0.5cm}, every three node
        part/.style={text width=0.5cm}, every four node
        part/.style={text width=0.5cm}, every five node
        part/.style={text width=0.5cm}, append after command={%
          \pgfextra{\let\mainnode=\tikzlastnode}%
          coordinate (c1 \mainnode) at
          ($(\mainnode.north)!.5!(\mainnode.one split)+(0.1cm,0)$)
          coordinate (c2 \mainnode) at ($(\mainnode.one
          split)!.5!(\mainnode.two split)+(0.1cm,0)$) coordinate (c3
          \mainnode) at ($(\mainnode.two split)!.5!(\mainnode.three
          split)+(0.1cm,0)$) coordinate (c4 \mainnode) at
          ($(\mainnode.three split)!.5!(\mainnode.four
          split)+(0.1cm,0)$) coordinate (c5 \mainnode) at
          ($(\mainnode.four split)!.5!(\mainnode.south)+(0.1cm,0)$) }
      } }

    \tikzset{ datanode/.style={%
        draw, rectangle, rectangle split, rectangle split parts=2,
        rectangle split part align=top,
        font=\scriptsize, inner sep=2pt, every one node
        part/.style={text width=1cm}, every two node part/.style={text
          width=1cm}, append after command={%
          \pgfextra{\let\mainnode=\tikzlastnode}%
          coordinate (c1 \mainnode) at
          ($(\mainnode.north)!.5!(\mainnode.one split)+(0.25cm,0)$)
          coordinate (c2 \mainnode) at ($(\mainnode.one
          split)!.5!(\mainnode.south)+(0.25cm,0)$) } } }

  \begin{tikzpicture}
    \tikzstyle{pointer} = [->, shorten >= 2pt] \tikzstyle{pointertail}
    = [fill,circle,minimum size=3pt,inner sep=0pt]
    \tikzstyle{nullpointer} = [draw,circle,minimum size=3pt,inner
    sep=0pt] \node[trienode] (root) {\nodepart{one}$1$ \nodepart{two}
      $2$ \nodepart{three} $3$ \nodepart{four} $4$ \nodepart{five}
      $5$};

    \node[datanode] (d1) at (1.5,1) {\nodepart{one} data
      \nodepart{two} trie};

    \node[trienode] (r2) at (3,1) {\nodepart{one}$1$ \nodepart{two}
      $2$ \nodepart{three} $3$ \nodepart{four} $4$ \nodepart{five}
      $5$};

    \node[datanode] (d2) at (4.5,1) {\nodepart{one} data
      \nodepart{two} trie};

    \node[trienode] (r3) at (6,1) {\nodepart{one}$1$ \nodepart{two}
      $2$ \nodepart{three} $3$ \nodepart{four} $4$ \nodepart{five}
      $5$};

    \node[datanode] (d3) at (7.5,1.5) {\nodepart{one} data
      \nodepart{two} trie};

    \node[datanode] (d4) at (7.5,0.5) {\nodepart{one} data
      \nodepart{two} trie};
    
    \node[datanode] (d6) at (1.5,-1) {\nodepart{one} data
      \nodepart{two} trie};

    \node[trienode] (r4) at (3,-1) {\nodepart{one}$1$ \nodepart{two}
      $2$ \nodepart{three} $3$ \nodepart{four} $4$ \nodepart{five}
      $5$};

    \node[trienode] (r5) at (6,-1) {\nodepart{one}$1$ \nodepart{two}
      $2$ \nodepart{three} $3$ \nodepart{four} $4$ \nodepart{five}
      $5$};

    \node[datanode] (d7) at (4.5,-1) {\nodepart{one} data
      \nodepart{two} trie};

    \node[datanode] (d8) at (7.5,-1) {\nodepart{one} data
      \nodepart{two} trie};

    \draw[pointer] (c1 root)node[pointertail]{}|-(d1.west);

    \draw[pointer] (c2 d1) node[pointertail]{}--++(right:5mm)
    |-(r2.west);

    \draw[pointer] (c2 r2)node[pointertail]{}--++(right:5mm)
    |-(d2.west);

    \draw[pointer] (c2 d2)node[pointertail]{}--++(right:5mm)
    |-(r3.west);

    \draw[pointer] (c3 r3)node[pointertail]{}--++(right:5mm)
    |-(d3.west);

    \draw[pointer] (c4 r3)node[pointertail]{}--++(right:5mm)
    |-(d4.west);

    \draw[pointer] (c2 root)node[pointertail]{}--++(right:5mm)
    |-(d6.west);

    \draw[pointer] (c2 d6)node[pointertail]{}--++(right:5mm)
    |-(r4.west);

    \draw[pointer] (c3 r4)node[pointertail]{}--++(right:5mm)
    |-(d7.west);

    \draw[pointer] (c2 d7)node[pointertail]{}--++(right:5mm)
    |-(r5.west);

    \draw[pointer] (c5 r5)node[pointertail]{}--++(right:5mm)
    |-(d8.west);

    \draw[pointer, shorten >=5pt] (c1 d2) node[pointertail]{} -- ($(c1
    d2)+(0,0.5cm)$) node {5};

    \draw[pointer, shorten >=5pt] (c1 d3) node[pointertail]{} -- ($(c1
    d3)+(1cm,0)$) node {10};

    \draw[pointer, shorten >=5pt] (c1 d4) node[pointertail]{} -- ($(c1
    d4)+(1cm,0)$) node {8};

    \draw[pointer, shorten >=5pt] (c1 d8) node[pointertail]{} -- ($(c1
    d8)+(1cm,0)$) node {2};
  \end{tikzpicture}
  \caption{A trie that associates integer values with sequences of
    integers in~$\{1,\dots,5\}$. That is, each node of the trie is an
    array of size five. With $(1,2)$ the trie associates~5,
    with~$(1,2,3)$ it associates~10, with~$(1,2,4)$ it associates~$8$,
    and, finally, with~$(2,3,5)$ it associates~$2$.}
  \label{fig:trie}
\end{figure}
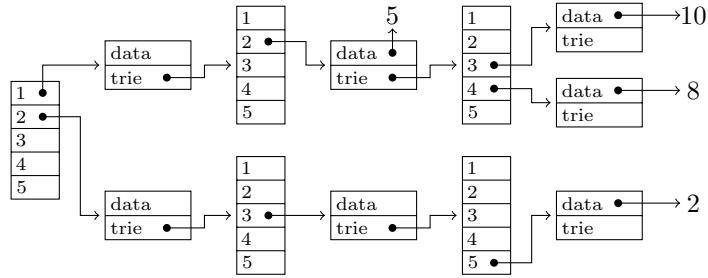

We define a \emph{trie} as a size-$\vernum{}$ array, of which each
cell contains a pointer to a structure consisting of two more
pointers: one of them points to data, the other one to another
trie. This is illustrated in \autoref{fig:trie}. A look-up of the
value associated with a sequence~$s=(s_1,\dots,s_d)$ in a trie~$T_1$
then works in $\bigO(d)$~time as follows: for $i\in[d-1]$, we get the
trie~$T_{i+1}$ pointed to by~$T_i[s_i]$. Then, from $T_{d}[s_{d}]$, we
get a pointer to the data associated with~$s$.

In the creation of the trie to associate values with the sequences
in~$L$, we face a problem: we do not have enough time to initialize all
cells of all arrays that implement the inner nodes of the trie: this
would take $\bigTh(\vernum{})$~time per node and, as seen in
\autoref{fig:trie}, the number of nodes required in the trie can be
more than~$|L|$.  This is a problem since, when creating the trie, we
do not know whether an array cell already contains a pointer to a trie
node of the next level or whether we have to create such a pointer
with a corresponding new node.  We have to make sure that we only
follow pointers of initialized cells and that we do not overwrite
previously correctly set up pointers since, otherwise, information in
subtries will be lost. We achieve this as follows:

The input list~$L$ contains sequences of length at most~$d$ of
integers in~$[n]$. Hence, we can sort~$L$ lexicographically in
$\bigO(d\cdot(\vernum{}+|L|))=\bigO(d\cdot\vernum{}+d\cdot|L|)$~time
using radix sort~\citep[Section~8.3]{CLRS01}. We construct the trie by
iterating over~$L$ once.  For each sequence~$p$ in~$L$, we find in
$\bigO(d)$~time the first position~$i$ in which the sequence~$p$
differs from its predecessor sequence in the lexicographically sorted
list~$L$.  This tells us that we already created all nodes on the path
from the trie's root node to the leaf corresponding to~$s$ up to a
depth of~$i$. Pointers up to this depth~$i$ are valid and may not be
overwritten. Nodes and pointers beyond this depth have to be newly
created.
\end{proof}

\noindent Using \autoref{lem:stringtrie}, we can finally prove that
\autoref{alg:eff-sunflo} runs in linear time.  Note that, together
with Propositions~\ref{lem:kern-correct} and~\ref{prop:small},
\autoref{lem:kern-schnell} completes the proof of
\autoref{thm:optkern-lintime}.

\begin{proposition}\label{lem:kern-schnell}
  \autoref{alg:eff-sunflo} can be implemented to run in
  $\bigO(d\cdot\vernum{}+2^dd\cdot \edgenum{})$ time.
\end{proposition}

\begin{proof}
  We first describe how \autoref{lem:stringtrie} helps us efficiently
  implementing the associative arrays~$\petcount[]$ and~$\markused[]$
  required by \autoref{alg:eff-sunflo}.  To this end, we assume that
  every vertex is represented as an integer in~$[n]$ and that every
  hyperedge is represented as a sequence sorted by increasing vertex
  numbers, which we call \emph{sorted hyperedge}.  We can initially
  sort each hyperedge of~$\HG$ in $\bigO(\edgenum{}\cdot d\log
  d)$~total time. Note that, on hyperedges represented as sorted
  sequences, the set subtraction operation needed in
  \autoref{lin:Hsuitable} can be executed in $\bigO(d)$~time such that
  the resulting set is again
  sorted~\citep[Section~4.4]{AHU83}. Moreover, we can generate all
  subsets of a sorted set such that the resulting subsets are
  sorted. Hence, we may assume to always deal with sorted hyperedges
  as a unique representation of hyperedges.

  We now apply \autoref{lem:stringtrie}.  Observe that
  \autoref{alg:eff-sunflo} looks up~$\petcount[C]$ and~$\markused[C]$
  only for sets~$C\subseteq e$ for some hyperedge~$e$.  Thus, from the
  set of sorted hyperedges, in $\bigO(2^dd\cdot\edgenum{})$~time, we
  compute a length-$(2^d\cdot\edgenum{})$ list~$L$ of all possible
  sets~$\Core\subseteq\Edg$ for all hyperedges~$\Edg$ and use this
  list in \autoref{lem:stringtrie} to create the associative arrays
  $\petcount[]$ and $\markused[]$ in $\bigO(d\cdot
  \vernum{}+d\cdot|L|)=\bigO(d\cdot\vernum+2^dd\cdot\edgenum{})$~time.

  Now, we can implement lines \ref{lin:initstart}--\ref{lin:initend}
  of \autoref{alg:eff-sunflo} to run in
  $\bigO(d\cdot\vernum{}+2^dd\cdot\edgenum{})$~time, observing that
  the loop in \autoref{lin:setfalse} can be implemented to run in
  $\bigO(d)$-time, since only one look-up to~$\markused[\Core]$ is
  needed to obtain a pointer to an array in which, then,
  $\bigO(d)$~values are set.

  The for-loop in \autoref{lin:eachedge} iterates
  $\edgenum{}$~times. Its body works in $\bigO(2^dd)$~time: obviously,
  this time bound holds for lines \ref{lin:necess} and \ref{lin:addH};
  it remains to show that the body of the for-loop in
  \autoref{lin:eachcore} works in $\bigO(d)$~time. This is easy to see
  if one considers that, in lines \ref{lin:Hsuitable} and
  \ref{lin:markused}, one only has to do one look-up to
  $\markused[\Core]$ to find a pointer to an array that holds the
  values for the at most $d$~vertices of a hyperedge. Also
  \autoref{lin:vertunion} works in linear time by first initializing
  all entries of an array~$\verts[]$ of size~$\vernum{}$ to
  ``$\bfalse$'' and then, for each output hyperedge~$\Edg$ and each
  vertex~$v\in\Edg$, setting ``$\verts[v]\gets\btrue$'' in
  $\bigO(d)$~time. Afterward, we can build the vertex set~$V'$ of the
  output hypergraph~$\HGout{}$ using the vertices~$v$ for which
  $\verts[v]=\btrue$. This takes
  $\bigO(\vernum{}+d\cdot\edgenum{})$~time.
\end{proof}

\section{Experimental evaluation}\label{sec:hsexp}
This section experimentally evaluates the linear-time kernelization
algorithm from \autoref{sec:expl}.  We demonstrate to which size our
algorithm can process instances within five minutes.

\paragraph{Implementation details.}
Our implementation of \autoref{alg:eff-sunflo} comprises about
700~lines of C++ and is freely
available.\footnote{\url{http://fpt.akt.tu-berlin.de/hslinkern/}} The
experiments were run on a computer with a 3.6\,GHz Intel Xeon
processor and 64\,GB RAM under Linux~3.2.0, where the C++ source code
has been compiled using the GNU C++ compiler in version~4.7.2 and
using the highest optimization level~(-O3).

Given a hypergraph~$\HG=(\Verts,\Edgs)$, our implementation of
\autoref{alg:eff-sunflo}
checks for each hyperedge~$\Edg\in\Edgs$ with $\ell:=|\Edg|$
independently whether it is a \emph{large
  hyperedge}~($2^\ell>\edgenum{}$) or a \emph{small hyperedge}
($2^\ell\leq\edgenum{}$). For a small hyperedge~$\Edg$,
\autoref{alg:eff-sunflo} chooses to consider all
subsets~$\Core\in\Edg$ as possible cores of sunflowers in
\autoref{lin:necess}. For a large hyperedge~$e$, all
subsets~$\Edg\cap\Edg'$ for any~$\Edg'\in\Edgs$ are considered as
possible cores instead. Hence, our implementation chooses the variant
which promises the lower running time for each hyperedge
independently.

Additionally to discarding hyperedges that contain some core of a
sunflower of size~$k+1$, our implementation of
\autoref{alg:eff-sunflo} also makes sure that the output hypergraph
contains no pair of hyperedges such that one is a superset of the
other.  To this end, the implementation initially sorts all hyperedges
by increasing cardinality in $\bigO(d+m)$~time using counting
sort~\citep[Section~8.2]{CLRS01}.  Moreover, after adding a
hyperedge~$\Edg{}$ to the output hypergraph, the implementation
sets~$\petcount[\Edg{}]$ to~$k+1$: in this way, the algorithm will not
add hyperedges to the output hypergraph that are supersets of already
added hyperedges.

As data structures to hold the values $\markused[C]$ and
$\petcount[C]$ used by \autoref{alg:eff-sunflo} to associate values
with sets~$C$ of size at most~$d$, we implemented the following
variants.
\begin{description}
\item[By \emph{malloc trie},] we refer to the associative array
  created in \autoref{lem:stringtrie}. It is implemented as a trie
  whose nodes are allocated as uninitialized arrays in constant time
  using the C-routine \texttt{malloc}. It guarantees
  $\bigO(d)$~look-up time and, as described in the proof of
  \autoref{lem:kern-schnell},
  $\bigO(d\cdot\vernum{}+2^dd\cdot\edgenum{})$~creation time. However,
  $\Omega(\vernum{}\cdot \edgenum{})$ random access memory cells may
  need to be reserved by the program in the worst case, although at
  most $\bigO(d\cdot\vernum{}+2^dd\cdot\edgenum{})$ memory cells are
  actually accessed.
\item[By \emph{calloc trie},] we refer to a trie whose nodes are
  allocated as arrays pre-initialized by zero using the C-routine
  \texttt{calloc}. This makes the intricate initialization by
  \autoref{lem:stringtrie} unnecessary. However, the running time of
  acquiring a zero-initialized array and its actual memory usage may
  vary depending on the implementation of the routine by the used
  C~library. For na\"{i}ve implementations of \texttt{calloc},
  creation time and memory usage of the calloc tree could be
  $\Omega(\vernum{}\cdot\edgenum{})$ in the worst case. However, we
  can still guarantee $\bigO(d)$~look-up time.
\item[By \emph{hash table},] we refer to an associative array
  implemented using the data structure~\texttt{unordered\_map}
  provided in the C++11 Standard Template Library. At most
  $\bigO(2^d\cdot\edgenum{})$~values are stored in the hash
  table. According to the C++11 reference, storage and look-up work in
  $\bigO(2^d\cdot\edgenum{})$~time in the worst case, but in
  $\bigO(d)$~time in the average case, where $\bigO(d)$~time accounts
  for computing the hash value of a hyperedge of cardinality~$d$.
\item[By \emph{balanced tree},] we refer to an associative array
  implemented using the \texttt{map} data structure provided by the
  C++11 Standard Template Library.  According to the C++11 reference,
  it is usually implemented as a balanced binary tree.  Since, in the
  worst case, $\bigO(2^d\cdot\edgenum{})$~values are stored in the
  tree, the C++11 reference guarantees $\bigO(d+\log\edgenum{})$~time
  for storage and look-up. Its memory requirements are
  $\bigO(2^d\edgenum{})$.
\end{description}

\paragraph{Data.} We execute our experiments on hypergraphs generated
from the \textsc{Golomb Subruler} problem: one gets as input a
set~$R\subseteq\mathbb N$ and wants to remove at most $k$~numbers
(``marks'') from~$R$ such that the result is a Golomb ruler, that is,
no pair of remaining marks has the same distance as another pair.  The
applications of Golomb rulers lie, among others, in radio frequency
allocation~\citep{Dra09}.  Optimum solutions for \textsc{Golomb
  Subruler} are only known for~$R=[n]$ with~$n\leq 553$ at the current
time.\footnote{\url{http://blogs.distributed.net/2014/02/}}

From a \textsc{Golomb Subruler} instance, we obtain a \emph{conflict
  hypergraph} as follows: the vertex set is~$R$, and for each
$a,b,c,d\in R$, create a hyperedge~$\{a,b,c,d\}$ if
$|a-b|=|c-d|$. Asking for a hitting set of size~$k$ in this conflict
hypergraph is equivalent to \textsc{Golomb Subruler} \citep{SMNW14}.
As shown by \citet{SMNW14}, the class of conflict hypergraphs for
$R=[n]$ has $n$~vertices and $\Theta(n^3)$~hyperedges, their
cardinality being three or four. Our data set consists of the conflict
hypergraphs for \textsc{Golomb Subruler} instances~$R=[n]$ with
$100\leq n\leq 600$, which yields conflict hypergraphs with $10^5$ to
$2\cdot 10^7$~hyperedges. Since, in this way, we obtain a whole family
of growing hypergraphs, this data set is well-suited to show the
running time and memory scalability of \autoref{alg:eff-sunflo}.

\paragraph{Experimental setup.} 
\autoref{alg:eff-sunflo} requires as input not only a
hypergraph~$\HG{}$ but also an upper bound~$k$ on the size of a sought
hitting set. We choose as~$k$ an upper bound on the size of a
\emph{minimum} hitting set, so that the kernelization algorithm will
not output small trivial no-instances and so that the computed problem
kernel will retain all minimum hitting sets.

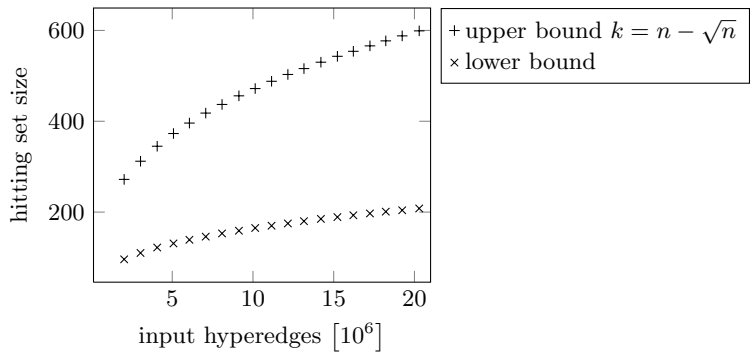
\begin{figure}\centering\small
  \begin{tikzpicture}
    \begin{axis}[ylabel={hitting set size}, xlabel={input hyperedges},
      x unit=10^6, change x base, axis base prefix={axis x base -6
        prefix {}}, legend cell align=left, legend pos=outer north
      east, width=0.5\textwidth,xmax=21000000]

      \addplot[color=black,mark=+,only marks] table [x=In, y=UBo, col
      sep=comma, skip coords between index={0}{19}]
      {FILE,In,UBo,LBo,Out,T1,T2,T3,T4,T5,Mem
hslinkern_balanced:gol106:1,100647,95,35,98792,0.88,0.87,0.88,0.88,0.9,123285504
hslinkern_balanced:gol134:1,202742,122,45,199772,1.98,1.98,1.99,1.99,2,233385984
hslinkern_balanced:gol153:1,301378,140,51,297502,3.12,3.15,3.13,3.12,3.13,340340736
hslinkern_balanced:gol169:1,405790,156,56,401058,4.4,4.38,4.41,4.38,4.39,452538368
hslinkern_balanced:gol182:1,506506,168,61,501016,5.61,5.61,5.63,5.59,5.61,561590272
hslinkern_balanced:gol193:1,603728,179,64,597552,6.87,6.87,6.87,6.83,6.84,666447872
hslinkern_balanced:gol204:1,712657,189,68,705755,8.17,8.24,8.24,8.23,8.21,784936960
hslinkern_balanced:gol213:1,810953,198,71,803427,9.53,9.46,9.51,9.59,9.61,890843136
hslinkern_balanced:gol222:1,917896,207,74,909719,10.85,10.97,10.89,10.91,10.82,1007235072
hslinkern_balanced:gol229:1,1007285,213,76,998583,12.16,12.05,12.04,12.04,12.02,1103704064
hslinkern_balanced:gol237:1,1116339,221,79,1107017,13.44,13.49,13.47,13.42,13.47,1222193152
hslinkern_balanced:gol244:1,1217987,228,81,1208105,14.94,15,14.98,15.06,15.03,1331245056
hslinkern_balanced:gol251:1,1325625,235,84,1315167,16.34,16.36,16.3,16.3,16.28,1448685568
hslinkern_balanced:gol257:1,1422784,240,86,1411819,17.75,17.74,17.83,17.65,17.81,1553543168
hslinkern_balanced:gol263:1,1524578,246,88,1513094,19.17,19.04,19.21,19.05,19.24,1663643648
hslinkern_balanced:gol269:1,1631115,252,90,1619100,20.55,20.64,20.72,20.45,20.68,1778987008
hslinkern_balanced:gol274:1,1723597,257,91,1711130,21.74,21.86,21.89,22.04,21.91,1878601728
hslinkern_balanced:gol279:1,1819510,262,93,1806583,23.11,23.2,23.21,23.14,23.15,1983459328
hslinkern_balanced:gol284:1,1918917,267,95,1905522,24.5,24.55,24.51,24.62,24.52,2090414080
hslinkern_balanced:gol289:1,2021880,272,96,2008008,26.15,26.04,26.11,26.19,26.04,2202611712
hslinkern_balanced:gol331:1,3035725,312,110,3017520,41.11,40.66,40.65,40.73,40.56,3300470784
hslinkern_balanced:gol365:1,4068883,345,122,4046740,56.37,56.4,56.2,56.27,56.11,4420349952
hslinkern_balanced:gol393:1,5077478,373,131,5051802,71.74,71.81,71.71,71.9,71.96,5512966144
hslinkern_balanced:gol417:1,6064344,396,139,6035432,87.34,87.33,87.27,86.83,87.25,6583562240
hslinkern_balanced:gol439:1,7074430,418,146,7042383,103.11,103.55,103.25,103.61,103.09,7678275584
hslinkern_balanced:gol459:1,8084845,437,153,8049808,120.56,120.24,120.73,119.72,119.98,8772988928
hslinkern_balanced:gol478:1,9129800,456,159,9091799,136.12,135.91,136.01,136.79,136.22,9906499584
hslinkern_balanced:gol495:1,10137868,472,165,10097113,154.21,154.27,153.41,153.35,153.07,10999115776
hslinkern_balanced:gol511:1,11152000,488,170,11108565,177.1,171.53,177.27,169.65,176.72,12099072000
hslinkern_balanced:gol526:1,12162172,503,175,12116147,186.9,190.63,187.16,191.04,187.27,13193785344
hslinkern_balanced:gol540:1,13158405,516,180,13109895,205.3,205.06,205.03,205.48,204.99,14274867200
hslinkern_balanced:gol554:1,14207607,530,185,14156547,223.16,226.4,223.38,225.03,222.48,15411523584
hslinkern_balanced:gol567:1,15230494,543,189,15177007,240.5,242.09,241.02,241.02,239.52,16521965568
hslinkern_balanced:gol579:1,16217235,554,193,16161458,256.97,258.12,257.45,257.75,257.35,17591513088
hslinkern_balanced:gol591:1,17245700,566,197,17187585,276.13,277.29,276.94,275.85,277.45,18707197952
hslinkern_balanced:gol602:1,18225851,577,201,18165551,293.66,295.5,293.94,294.82,293.32,19770454016
hslinkern_balanced:gol613:1,19242453,588,204,19179927,311.34,311.62,311.05,312.02,311.28,20872507392
hslinkern_balanced:gol624:1,20296172,599,208,20231380,330.55,331.72,331.73,330.45,331.39,22015455232
}; \addlegendentry{upper bound $k=n-\sqrt n$};

      \addplot[color=black,mark=x,only marks] table [x=In, y=LBo, col
      sep=comma, skip coords between index={0}{19}]
      {FILE,In,UBo,LBo,Out,T1,T2,T3,T4,T5,Mem
hslinkern_balanced:gol106:1,100647,95,35,98792,0.88,0.87,0.88,0.88,0.9,123285504
hslinkern_balanced:gol134:1,202742,122,45,199772,1.98,1.98,1.99,1.99,2,233385984
hslinkern_balanced:gol153:1,301378,140,51,297502,3.12,3.15,3.13,3.12,3.13,340340736
hslinkern_balanced:gol169:1,405790,156,56,401058,4.4,4.38,4.41,4.38,4.39,452538368
hslinkern_balanced:gol182:1,506506,168,61,501016,5.61,5.61,5.63,5.59,5.61,561590272
hslinkern_balanced:gol193:1,603728,179,64,597552,6.87,6.87,6.87,6.83,6.84,666447872
hslinkern_balanced:gol204:1,712657,189,68,705755,8.17,8.24,8.24,8.23,8.21,784936960
hslinkern_balanced:gol213:1,810953,198,71,803427,9.53,9.46,9.51,9.59,9.61,890843136
hslinkern_balanced:gol222:1,917896,207,74,909719,10.85,10.97,10.89,10.91,10.82,1007235072
hslinkern_balanced:gol229:1,1007285,213,76,998583,12.16,12.05,12.04,12.04,12.02,1103704064
hslinkern_balanced:gol237:1,1116339,221,79,1107017,13.44,13.49,13.47,13.42,13.47,1222193152
hslinkern_balanced:gol244:1,1217987,228,81,1208105,14.94,15,14.98,15.06,15.03,1331245056
hslinkern_balanced:gol251:1,1325625,235,84,1315167,16.34,16.36,16.3,16.3,16.28,1448685568
hslinkern_balanced:gol257:1,1422784,240,86,1411819,17.75,17.74,17.83,17.65,17.81,1553543168
hslinkern_balanced:gol263:1,1524578,246,88,1513094,19.17,19.04,19.21,19.05,19.24,1663643648
hslinkern_balanced:gol269:1,1631115,252,90,1619100,20.55,20.64,20.72,20.45,20.68,1778987008
hslinkern_balanced:gol274:1,1723597,257,91,1711130,21.74,21.86,21.89,22.04,21.91,1878601728
hslinkern_balanced:gol279:1,1819510,262,93,1806583,23.11,23.2,23.21,23.14,23.15,1983459328
hslinkern_balanced:gol284:1,1918917,267,95,1905522,24.5,24.55,24.51,24.62,24.52,2090414080
hslinkern_balanced:gol289:1,2021880,272,96,2008008,26.15,26.04,26.11,26.19,26.04,2202611712
hslinkern_balanced:gol331:1,3035725,312,110,3017520,41.11,40.66,40.65,40.73,40.56,3300470784
hslinkern_balanced:gol365:1,4068883,345,122,4046740,56.37,56.4,56.2,56.27,56.11,4420349952
hslinkern_balanced:gol393:1,5077478,373,131,5051802,71.74,71.81,71.71,71.9,71.96,5512966144
hslinkern_balanced:gol417:1,6064344,396,139,6035432,87.34,87.33,87.27,86.83,87.25,6583562240
hslinkern_balanced:gol439:1,7074430,418,146,7042383,103.11,103.55,103.25,103.61,103.09,7678275584
hslinkern_balanced:gol459:1,8084845,437,153,8049808,120.56,120.24,120.73,119.72,119.98,8772988928
hslinkern_balanced:gol478:1,9129800,456,159,9091799,136.12,135.91,136.01,136.79,136.22,9906499584
hslinkern_balanced:gol495:1,10137868,472,165,10097113,154.21,154.27,153.41,153.35,153.07,10999115776
hslinkern_balanced:gol511:1,11152000,488,170,11108565,177.1,171.53,177.27,169.65,176.72,12099072000
hslinkern_balanced:gol526:1,12162172,503,175,12116147,186.9,190.63,187.16,191.04,187.27,13193785344
hslinkern_balanced:gol540:1,13158405,516,180,13109895,205.3,205.06,205.03,205.48,204.99,14274867200
hslinkern_balanced:gol554:1,14207607,530,185,14156547,223.16,226.4,223.38,225.03,222.48,15411523584
hslinkern_balanced:gol567:1,15230494,543,189,15177007,240.5,242.09,241.02,241.02,239.52,16521965568
hslinkern_balanced:gol579:1,16217235,554,193,16161458,256.97,258.12,257.45,257.75,257.35,17591513088
hslinkern_balanced:gol591:1,17245700,566,197,17187585,276.13,277.29,276.94,275.85,277.45,18707197952
hslinkern_balanced:gol602:1,18225851,577,201,18165551,293.66,295.5,293.94,294.82,293.32,19770454016
hslinkern_balanced:gol613:1,19242453,588,204,19179927,311.34,311.62,311.05,312.02,311.28,20872507392
hslinkern_balanced:gol624:1,20296172,599,208,20231380,330.55,331.72,331.73,330.45,331.39,22015455232
}; \addlegendentry{lower bound};
    \end{axis}
  \end{tikzpicture}
  \caption{Upper and lower bounds for the minimum hitting set sizes
    for the data set obtained from the \textsc{Golomb Subruler}
    problem. The number of vertices~$n$ in the input instances is
    omitted, as it almost coincides with the upper bound~$k=n-\sqrt n$
    of the hitting set size. The lower bound was obtained from a
    maximal set of pairwise disjoint hyperedges.}
  \label{fig:hs-bounds}
\end{figure}

To obtain this upper bound for the \hs4 instances that we obtain from
\textsc{Golomb Subruler}, we exploit that, for all $n\leq 4.2\cdot
10^9$, \citet[Theorem~6.1]{Dim02} verified that there is a Golomb
ruler~$R\subseteq[n]$ with strictly more than $\sqrt{n}$~marks. Hence,
in our experiments with $n\leq 600$, a conflict hypergraph of a
\textsc{Golomb Subruler} instance~$R=[n]$ has a hitting set of size at
most~$k:=\lfloor n-\sqrt{n}\rfloor$. We use this~$k$ to compute
problem kernels for \hs 4.  \autoref{fig:hs-bounds} shows this upper
bound together with a lower bound.

\paragraph{Experimental results.}
In all plots to be shown, each point has been obtained from a single
run of our algorithm; the running times and memory usage are not
averaged in any way.

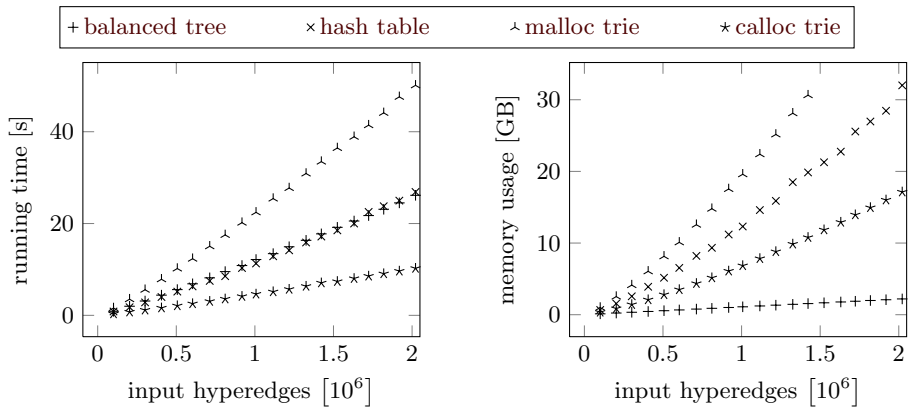
\begin{figure}\small\center
  \ref{legends}
  \begin{tikzpicture}
    \begin{axis}[axis base prefix={axis x base -6 prefix {}}, y
      unit=s, x unit=10^6, xlabel=input hyperedges, ylabel=running
      time, width=0.5\textwidth, legend cell align=left, legend
      pos=north west, change x base, axis base prefix={axis x base -6
        prefix {}}, xmax=2050000, legend to name=legends, legend
      columns=-1, legend style={/tikz/every even column/.append
        style={column sep=1cm}}]
      
      \addplot[color=black,mark=+, only marks] table [x=In, y=T1, col
      sep=comma] {FILE,In,UBo,LBo,Out,T1,T2,T3,T4,T5,Mem
hslinkern_balanced:gol106:1,100647,95,35,98792,0.88,0.87,0.88,0.88,0.9,123285504
hslinkern_balanced:gol134:1,202742,122,45,199772,1.98,1.98,1.99,1.99,2,233385984
hslinkern_balanced:gol153:1,301378,140,51,297502,3.12,3.15,3.13,3.12,3.13,340340736
hslinkern_balanced:gol169:1,405790,156,56,401058,4.4,4.38,4.41,4.38,4.39,452538368
hslinkern_balanced:gol182:1,506506,168,61,501016,5.61,5.61,5.63,5.59,5.61,561590272
hslinkern_balanced:gol193:1,603728,179,64,597552,6.87,6.87,6.87,6.83,6.84,666447872
hslinkern_balanced:gol204:1,712657,189,68,705755,8.17,8.24,8.24,8.23,8.21,784936960
hslinkern_balanced:gol213:1,810953,198,71,803427,9.53,9.46,9.51,9.59,9.61,890843136
hslinkern_balanced:gol222:1,917896,207,74,909719,10.85,10.97,10.89,10.91,10.82,1007235072
hslinkern_balanced:gol229:1,1007285,213,76,998583,12.16,12.05,12.04,12.04,12.02,1103704064
hslinkern_balanced:gol237:1,1116339,221,79,1107017,13.44,13.49,13.47,13.42,13.47,1222193152
hslinkern_balanced:gol244:1,1217987,228,81,1208105,14.94,15,14.98,15.06,15.03,1331245056
hslinkern_balanced:gol251:1,1325625,235,84,1315167,16.34,16.36,16.3,16.3,16.28,1448685568
hslinkern_balanced:gol257:1,1422784,240,86,1411819,17.75,17.74,17.83,17.65,17.81,1553543168
hslinkern_balanced:gol263:1,1524578,246,88,1513094,19.17,19.04,19.21,19.05,19.24,1663643648
hslinkern_balanced:gol269:1,1631115,252,90,1619100,20.55,20.64,20.72,20.45,20.68,1778987008
hslinkern_balanced:gol274:1,1723597,257,91,1711130,21.74,21.86,21.89,22.04,21.91,1878601728
hslinkern_balanced:gol279:1,1819510,262,93,1806583,23.11,23.2,23.21,23.14,23.15,1983459328
hslinkern_balanced:gol284:1,1918917,267,95,1905522,24.5,24.55,24.51,24.62,24.52,2090414080
hslinkern_balanced:gol289:1,2021880,272,96,2008008,26.15,26.04,26.11,26.19,26.04,2202611712
hslinkern_balanced:gol331:1,3035725,312,110,3017520,41.11,40.66,40.65,40.73,40.56,3300470784
hslinkern_balanced:gol365:1,4068883,345,122,4046740,56.37,56.4,56.2,56.27,56.11,4420349952
hslinkern_balanced:gol393:1,5077478,373,131,5051802,71.74,71.81,71.71,71.9,71.96,5512966144
hslinkern_balanced:gol417:1,6064344,396,139,6035432,87.34,87.33,87.27,86.83,87.25,6583562240
hslinkern_balanced:gol439:1,7074430,418,146,7042383,103.11,103.55,103.25,103.61,103.09,7678275584
hslinkern_balanced:gol459:1,8084845,437,153,8049808,120.56,120.24,120.73,119.72,119.98,8772988928
hslinkern_balanced:gol478:1,9129800,456,159,9091799,136.12,135.91,136.01,136.79,136.22,9906499584
hslinkern_balanced:gol495:1,10137868,472,165,10097113,154.21,154.27,153.41,153.35,153.07,10999115776
hslinkern_balanced:gol511:1,11152000,488,170,11108565,177.1,171.53,177.27,169.65,176.72,12099072000
hslinkern_balanced:gol526:1,12162172,503,175,12116147,186.9,190.63,187.16,191.04,187.27,13193785344
hslinkern_balanced:gol540:1,13158405,516,180,13109895,205.3,205.06,205.03,205.48,204.99,14274867200
hslinkern_balanced:gol554:1,14207607,530,185,14156547,223.16,226.4,223.38,225.03,222.48,15411523584
hslinkern_balanced:gol567:1,15230494,543,189,15177007,240.5,242.09,241.02,241.02,239.52,16521965568
hslinkern_balanced:gol579:1,16217235,554,193,16161458,256.97,258.12,257.45,257.75,257.35,17591513088
hslinkern_balanced:gol591:1,17245700,566,197,17187585,276.13,277.29,276.94,275.85,277.45,18707197952
hslinkern_balanced:gol602:1,18225851,577,201,18165551,293.66,295.5,293.94,294.82,293.32,19770454016
hslinkern_balanced:gol613:1,19242453,588,204,19179927,311.34,311.62,311.05,312.02,311.28,20872507392
hslinkern_balanced:gol624:1,20296172,599,208,20231380,330.55,331.72,331.73,330.45,331.39,22015455232
}; \addlegendentry{balanced tree};

      \addplot[color=black,mark=x, only marks] table [x=In, y=T1, col
      sep=comma] {FILE,In,UBo,LBo,Out,T1,T2,T3,T4,T5,Mem
hslinkern_hash:gol106:1,100647,95,35,98792,0.84,0.81,0.82,0.82,0.82,659800064
hslinkern_hash:gol134:1,202742,122,45,199772,1.77,1.74,1.75,1.74,1.74,1557377024
hslinkern_hash:gol153:1,301378,140,51,297502,2.82,2.8,2.8,2.8,2.8,2567528448
hslinkern_hash:gol169:1,405790,156,56,401058,3.99,4,3.99,3.99,3.99,3900268544
hslinkern_hash:gol182:1,506506,168,61,501016,5.17,5.24,5.19,5.19,5.17,5156462592
hslinkern_hash:gol193:1,603728,179,64,597552,6.3,6.29,6.25,6.24,6.3,6534451200
hslinkern_hash:gol204:1,712657,189,68,705755,7.5,7.48,7.49,7.49,7.54,8207978496
hslinkern_hash:gol213:1,810953,198,71,803427,8.53,8.59,8.53,8.51,8.58,9337135104
hslinkern_hash:gol222:1,917896,207,74,909719,10.31,10.29,10.28,10.28,10.28,11193004032
hslinkern_hash:gol229:1,1007285,213,76,998583,11.3,11.39,11.29,11.28,11.31,12311834624
hslinkern_hash:gol237:1,1116339,221,79,1107017,12.86,12.87,12.87,12.84,12.91,14602973184
hslinkern_hash:gol244:1,1217987,228,81,1208105,14.16,14.18,14.16,14.17,14.14,15886540800
hslinkern_hash:gol251:1,1325625,235,84,1315167,15.87,15.87,16,15.89,15.88,18498543616
hslinkern_hash:gol257:1,1422784,240,86,1411819,17.16,17.21,17.25,17.25,17.12,19853303808
hslinkern_hash:gol263:1,1524578,246,88,1513094,18.51,18.57,18.6,18.47,18.46,21273075712
hslinkern_hash:gol269:1,1631115,252,90,1619100,20,19.97,20,19.92,20.15,22758907904
hslinkern_hash:gol274:1,1723597,257,91,1711130,22.53,22.61,22.43,22.52,22.47,25563832320
hslinkern_hash:gol279:1,1819510,262,93,1806583,23.77,23.63,23.64,23.64,23.62,26962632704
hslinkern_hash:gol284:1,1918917,267,95,1905522,24.99,24.99,25,24.93,24.91,28462096384
hslinkern_hash:gol289:1,2021880,272,96,2008008,26.88,26.86,26.9,26.88,26.89,32006283264
hslinkern_hash:gol331:1,3035725,312,110,3017520,42.79,42.85,42.9,42.76,42.77,33464852480
}; \addlegendentry{hash table};

      \addplot[color=black,mark=Mercedes star, only marks] table
      [x=In, y=T1, col sep=comma] {FILE,In,UBo,LBo,Out,T1,T2,T3,T4,T5,Mem
hslinkern_malloc:gol106:1,100647,95,35,98792,1.59,1.6,1.6,1.58,1.6,997797888
hslinkern_malloc:gol134:1,202742,122,45,199772,3.51,3.48,3.59,3.47,3.55,2526621696
hslinkern_malloc:gol153:1,301378,140,51,297502,5.57,5.49,5.51,5.57,5.49,4127797248
hslinkern_malloc:gol169:1,405790,156,56,401058,7.85,7.79,7.82,7.76,7.78,6047739904
hslinkern_malloc:gol182:1,506506,168,61,501016,10.19,10.17,10.19,10.16,10.3,8142794752
hslinkern_malloc:gol193:1,603728,179,64,597552,12.33,12.35,12.39,12.43,12.33,10094194688
hslinkern_malloc:gol204:1,712657,189,68,705755,15.04,15.1,15.07,15.08,15.03,12554153984
hslinkern_malloc:gol213:1,810953,198,71,803427,17.51,17.44,17.52,17.4,17.43,14788669440
hslinkern_malloc:gol222:1,917896,207,74,909719,20.18,20.18,20.25,20.48,20.22,17540132864
hslinkern_malloc:gol229:1,1007285,213,76,998583,22.43,22.41,22.43,22.72,22.41,19593244672
hslinkern_malloc:gol237:1,1116339,221,79,1107017,25.43,25.25,25.26,25.25,25.21,22389796864
hslinkern_malloc:gol244:1,1217987,228,81,1208105,27.73,27.93,27.65,27.72,27.75,25151746048
hslinkern_malloc:gol251:1,1325625,235,84,1315167,30.89,31.01,31.16,31.13,31.18,28109778944
hslinkern_malloc:gol257:1,1422784,240,86,1411819,33.48,33.51,33.61,33.67,33.52,30630555648
hslinkern_malloc:gol263:1,1524578,246,88,1513094,36.45,36.13,36.05,36.24,36.04,
hslinkern_malloc:gol269:1,1631115,252,90,1619100,38.89,39.02,38.82,38.9,38.78,
hslinkern_malloc:gol274:1,1723597,257,91,1711130,41.42,41.24,41.44,41.76,41.29,
hslinkern_malloc:gol279:1,1819510,262,93,1806583,44.1,44.07,44.03,44.4,43.98,
hslinkern_malloc:gol284:1,1918917,267,95,1905522,47.59,46.91,46.73,46.75,46.82,
hslinkern_malloc:gol289:1,2021880,272,96,2008008,50.1,49.52,49.62,49.95,49.59,
};
      \addlegendentry{malloc trie};

      \addplot[color=black,mark=star, only marks] table [x=In, y=T1,
      col sep=comma] {FILE,In,UBo,LBo,Out,T1,T2,T3,T4,T5,Mem
hslinkern_calloc:gol106:1,100647,95,35,98792,0.33,0.34,0.32,0.34,0.34,355016704
hslinkern_calloc:gol134:1,202742,122,45,199772,0.73,0.74,0.74,0.75,0.73,858333184
hslinkern_calloc:gol153:1,301378,140,51,297502,1.15,1.16,1.16,1.16,1.15,1418272768
hslinkern_calloc:gol169:1,405790,156,56,401058,1.62,1.63,1.62,1.62,1.62,2086215680
hslinkern_calloc:gol182:1,506506,168,61,501016,2.09,2.11,2.09,2.09,2.1,2802393088
hslinkern_calloc:gol193:1,603728,179,64,597552,2.54,2.54,2.57,2.54,2.56,3501793280
hslinkern_calloc:gol204:1,712657,189,68,705755,3.05,3.09,3.1,3.07,3.08,4359528448
hslinkern_calloc:gol213:1,810953,198,71,803427,3.57,3.59,3.59,3.57,3.56,5153300480
hslinkern_calloc:gol222:1,917896,207,74,909719,4.13,4.17,4.12,4.12,4.13,6094921728
hslinkern_calloc:gol229:1,1007285,213,76,998583,4.63,4.59,4.59,4.59,4.59,6848847872
hslinkern_calloc:gol237:1,1116339,221,79,1107017,5.14,5.14,5.17,5.18,5.17,7838703616
hslinkern_calloc:gol244:1,1217987,228,81,1208105,5.68,5.73,5.69,5.68,5.67,8804442112
hslinkern_calloc:gol251:1,1325625,235,84,1315167,6.34,6.35,6.4,6.26,6.28,9857212416
hslinkern_calloc:gol257:1,1422784,240,86,1411819,7.07,6.83,6.85,6.81,6.85,10783105024
hslinkern_calloc:gol263:1,1524578,246,88,1513094,7.38,7.44,7.43,7.38,7.46,11846361088
hslinkern_calloc:gol269:1,1631115,252,90,1619100,8.03,8.11,8.02,7.99,7.99,12907520000
hslinkern_calloc:gol274:1,1723597,257,91,1711130,8.55,8.71,8.52,8.57,8.53,13942464512
hslinkern_calloc:gol279:1,1819510,262,93,1806583,9.07,9.05,9.2,9.04,9.02,14950146048
hslinkern_calloc:gol284:1,1918917,267,95,1905522,9.64,9.65,9.65,9.65,9.63,16011304960
hslinkern_calloc:gol289:1,2021880,272,96,2008008,10.25,10.21,10.17,10.19,10.2,17129086976
hslinkern_calloc:gol331:1,3035725,312,110,3017520,16.27,16.46,16.26,16.29,16.25,29353385984
hslinkern_calloc:gol365:1,4068883,345,122,4046740,23.21,22.97,22.89,22.9,22.89,33466949632
hslinkern_calloc:gol393:1,5077478,373,131,5051802,29.43,29.61,29.43,29.64,29.77,33466949632
}; \addlegendentry{calloc trie};
    \end{axis}
  \end{tikzpicture}\hfill{}
  \begin{tikzpicture}
    \begin{axis}[change y base, change x base, axis base prefix={axis
        x base -6 prefix {}}, use units, y SI prefix=giga, y unit=B, x
      unit=10^6, xlabel=input hyperedges, ylabel=memory usage,
      width=0.5\textwidth, legend cell align=left, legend pos=north
      west, xmax=2050000]

      \addplot[color=black,mark=+,only marks] table [x=In, y=Mem, col
      sep=comma] {FILE,In,UBo,LBo,Out,T1,T2,T3,T4,T5,Mem
hslinkern_balanced:gol106:1,100647,95,35,98792,0.88,0.87,0.88,0.88,0.9,123285504
hslinkern_balanced:gol134:1,202742,122,45,199772,1.98,1.98,1.99,1.99,2,233385984
hslinkern_balanced:gol153:1,301378,140,51,297502,3.12,3.15,3.13,3.12,3.13,340340736
hslinkern_balanced:gol169:1,405790,156,56,401058,4.4,4.38,4.41,4.38,4.39,452538368
hslinkern_balanced:gol182:1,506506,168,61,501016,5.61,5.61,5.63,5.59,5.61,561590272
hslinkern_balanced:gol193:1,603728,179,64,597552,6.87,6.87,6.87,6.83,6.84,666447872
hslinkern_balanced:gol204:1,712657,189,68,705755,8.17,8.24,8.24,8.23,8.21,784936960
hslinkern_balanced:gol213:1,810953,198,71,803427,9.53,9.46,9.51,9.59,9.61,890843136
hslinkern_balanced:gol222:1,917896,207,74,909719,10.85,10.97,10.89,10.91,10.82,1007235072
hslinkern_balanced:gol229:1,1007285,213,76,998583,12.16,12.05,12.04,12.04,12.02,1103704064
hslinkern_balanced:gol237:1,1116339,221,79,1107017,13.44,13.49,13.47,13.42,13.47,1222193152
hslinkern_balanced:gol244:1,1217987,228,81,1208105,14.94,15,14.98,15.06,15.03,1331245056
hslinkern_balanced:gol251:1,1325625,235,84,1315167,16.34,16.36,16.3,16.3,16.28,1448685568
hslinkern_balanced:gol257:1,1422784,240,86,1411819,17.75,17.74,17.83,17.65,17.81,1553543168
hslinkern_balanced:gol263:1,1524578,246,88,1513094,19.17,19.04,19.21,19.05,19.24,1663643648
hslinkern_balanced:gol269:1,1631115,252,90,1619100,20.55,20.64,20.72,20.45,20.68,1778987008
hslinkern_balanced:gol274:1,1723597,257,91,1711130,21.74,21.86,21.89,22.04,21.91,1878601728
hslinkern_balanced:gol279:1,1819510,262,93,1806583,23.11,23.2,23.21,23.14,23.15,1983459328
hslinkern_balanced:gol284:1,1918917,267,95,1905522,24.5,24.55,24.51,24.62,24.52,2090414080
hslinkern_balanced:gol289:1,2021880,272,96,2008008,26.15,26.04,26.11,26.19,26.04,2202611712
hslinkern_balanced:gol331:1,3035725,312,110,3017520,41.11,40.66,40.65,40.73,40.56,3300470784
hslinkern_balanced:gol365:1,4068883,345,122,4046740,56.37,56.4,56.2,56.27,56.11,4420349952
hslinkern_balanced:gol393:1,5077478,373,131,5051802,71.74,71.81,71.71,71.9,71.96,5512966144
hslinkern_balanced:gol417:1,6064344,396,139,6035432,87.34,87.33,87.27,86.83,87.25,6583562240
hslinkern_balanced:gol439:1,7074430,418,146,7042383,103.11,103.55,103.25,103.61,103.09,7678275584
hslinkern_balanced:gol459:1,8084845,437,153,8049808,120.56,120.24,120.73,119.72,119.98,8772988928
hslinkern_balanced:gol478:1,9129800,456,159,9091799,136.12,135.91,136.01,136.79,136.22,9906499584
hslinkern_balanced:gol495:1,10137868,472,165,10097113,154.21,154.27,153.41,153.35,153.07,10999115776
hslinkern_balanced:gol511:1,11152000,488,170,11108565,177.1,171.53,177.27,169.65,176.72,12099072000
hslinkern_balanced:gol526:1,12162172,503,175,12116147,186.9,190.63,187.16,191.04,187.27,13193785344
hslinkern_balanced:gol540:1,13158405,516,180,13109895,205.3,205.06,205.03,205.48,204.99,14274867200
hslinkern_balanced:gol554:1,14207607,530,185,14156547,223.16,226.4,223.38,225.03,222.48,15411523584
hslinkern_balanced:gol567:1,15230494,543,189,15177007,240.5,242.09,241.02,241.02,239.52,16521965568
hslinkern_balanced:gol579:1,16217235,554,193,16161458,256.97,258.12,257.45,257.75,257.35,17591513088
hslinkern_balanced:gol591:1,17245700,566,197,17187585,276.13,277.29,276.94,275.85,277.45,18707197952
hslinkern_balanced:gol602:1,18225851,577,201,18165551,293.66,295.5,293.94,294.82,293.32,19770454016
hslinkern_balanced:gol613:1,19242453,588,204,19179927,311.34,311.62,311.05,312.02,311.28,20872507392
hslinkern_balanced:gol624:1,20296172,599,208,20231380,330.55,331.72,331.73,330.45,331.39,22015455232
};

      \addplot[color=black,mark=x, only marks] table [x=In, y=Mem, col
      sep=comma] {FILE,In,UBo,LBo,Out,T1,T2,T3,T4,T5,Mem
hslinkern_hash:gol106:1,100647,95,35,98792,0.84,0.81,0.82,0.82,0.82,659800064
hslinkern_hash:gol134:1,202742,122,45,199772,1.77,1.74,1.75,1.74,1.74,1557377024
hslinkern_hash:gol153:1,301378,140,51,297502,2.82,2.8,2.8,2.8,2.8,2567528448
hslinkern_hash:gol169:1,405790,156,56,401058,3.99,4,3.99,3.99,3.99,3900268544
hslinkern_hash:gol182:1,506506,168,61,501016,5.17,5.24,5.19,5.19,5.17,5156462592
hslinkern_hash:gol193:1,603728,179,64,597552,6.3,6.29,6.25,6.24,6.3,6534451200
hslinkern_hash:gol204:1,712657,189,68,705755,7.5,7.48,7.49,7.49,7.54,8207978496
hslinkern_hash:gol213:1,810953,198,71,803427,8.53,8.59,8.53,8.51,8.58,9337135104
hslinkern_hash:gol222:1,917896,207,74,909719,10.31,10.29,10.28,10.28,10.28,11193004032
hslinkern_hash:gol229:1,1007285,213,76,998583,11.3,11.39,11.29,11.28,11.31,12311834624
hslinkern_hash:gol237:1,1116339,221,79,1107017,12.86,12.87,12.87,12.84,12.91,14602973184
hslinkern_hash:gol244:1,1217987,228,81,1208105,14.16,14.18,14.16,14.17,14.14,15886540800
hslinkern_hash:gol251:1,1325625,235,84,1315167,15.87,15.87,16,15.89,15.88,18498543616
hslinkern_hash:gol257:1,1422784,240,86,1411819,17.16,17.21,17.25,17.25,17.12,19853303808
hslinkern_hash:gol263:1,1524578,246,88,1513094,18.51,18.57,18.6,18.47,18.46,21273075712
hslinkern_hash:gol269:1,1631115,252,90,1619100,20,19.97,20,19.92,20.15,22758907904
hslinkern_hash:gol274:1,1723597,257,91,1711130,22.53,22.61,22.43,22.52,22.47,25563832320
hslinkern_hash:gol279:1,1819510,262,93,1806583,23.77,23.63,23.64,23.64,23.62,26962632704
hslinkern_hash:gol284:1,1918917,267,95,1905522,24.99,24.99,25,24.93,24.91,28462096384
hslinkern_hash:gol289:1,2021880,272,96,2008008,26.88,26.86,26.9,26.88,26.89,32006283264
hslinkern_hash:gol331:1,3035725,312,110,3017520,42.79,42.85,42.9,42.76,42.77,33464852480
};

      \addplot[color=black,mark=Mercedes star, only marks] table
      [x=In, y=Mem, col sep=comma] {FILE,In,UBo,LBo,Out,T1,T2,T3,T4,T5,Mem
hslinkern_malloc:gol106:1,100647,95,35,98792,1.59,1.6,1.6,1.58,1.6,997797888
hslinkern_malloc:gol134:1,202742,122,45,199772,3.51,3.48,3.59,3.47,3.55,2526621696
hslinkern_malloc:gol153:1,301378,140,51,297502,5.57,5.49,5.51,5.57,5.49,4127797248
hslinkern_malloc:gol169:1,405790,156,56,401058,7.85,7.79,7.82,7.76,7.78,6047739904
hslinkern_malloc:gol182:1,506506,168,61,501016,10.19,10.17,10.19,10.16,10.3,8142794752
hslinkern_malloc:gol193:1,603728,179,64,597552,12.33,12.35,12.39,12.43,12.33,10094194688
hslinkern_malloc:gol204:1,712657,189,68,705755,15.04,15.1,15.07,15.08,15.03,12554153984
hslinkern_malloc:gol213:1,810953,198,71,803427,17.51,17.44,17.52,17.4,17.43,14788669440
hslinkern_malloc:gol222:1,917896,207,74,909719,20.18,20.18,20.25,20.48,20.22,17540132864
hslinkern_malloc:gol229:1,1007285,213,76,998583,22.43,22.41,22.43,22.72,22.41,19593244672
hslinkern_malloc:gol237:1,1116339,221,79,1107017,25.43,25.25,25.26,25.25,25.21,22389796864
hslinkern_malloc:gol244:1,1217987,228,81,1208105,27.73,27.93,27.65,27.72,27.75,25151746048
hslinkern_malloc:gol251:1,1325625,235,84,1315167,30.89,31.01,31.16,31.13,31.18,28109778944
hslinkern_malloc:gol257:1,1422784,240,86,1411819,33.48,33.51,33.61,33.67,33.52,30630555648
hslinkern_malloc:gol263:1,1524578,246,88,1513094,36.45,36.13,36.05,36.24,36.04,
hslinkern_malloc:gol269:1,1631115,252,90,1619100,38.89,39.02,38.82,38.9,38.78,
hslinkern_malloc:gol274:1,1723597,257,91,1711130,41.42,41.24,41.44,41.76,41.29,
hslinkern_malloc:gol279:1,1819510,262,93,1806583,44.1,44.07,44.03,44.4,43.98,
hslinkern_malloc:gol284:1,1918917,267,95,1905522,47.59,46.91,46.73,46.75,46.82,
hslinkern_malloc:gol289:1,2021880,272,96,2008008,50.1,49.52,49.62,49.95,49.59,
};

      \addplot[color=black,mark=star, only marks] table [x=In, y=Mem,
      col sep=comma] {FILE,In,UBo,LBo,Out,T1,T2,T3,T4,T5,Mem
hslinkern_calloc:gol106:1,100647,95,35,98792,0.33,0.34,0.32,0.34,0.34,355016704
hslinkern_calloc:gol134:1,202742,122,45,199772,0.73,0.74,0.74,0.75,0.73,858333184
hslinkern_calloc:gol153:1,301378,140,51,297502,1.15,1.16,1.16,1.16,1.15,1418272768
hslinkern_calloc:gol169:1,405790,156,56,401058,1.62,1.63,1.62,1.62,1.62,2086215680
hslinkern_calloc:gol182:1,506506,168,61,501016,2.09,2.11,2.09,2.09,2.1,2802393088
hslinkern_calloc:gol193:1,603728,179,64,597552,2.54,2.54,2.57,2.54,2.56,3501793280
hslinkern_calloc:gol204:1,712657,189,68,705755,3.05,3.09,3.1,3.07,3.08,4359528448
hslinkern_calloc:gol213:1,810953,198,71,803427,3.57,3.59,3.59,3.57,3.56,5153300480
hslinkern_calloc:gol222:1,917896,207,74,909719,4.13,4.17,4.12,4.12,4.13,6094921728
hslinkern_calloc:gol229:1,1007285,213,76,998583,4.63,4.59,4.59,4.59,4.59,6848847872
hslinkern_calloc:gol237:1,1116339,221,79,1107017,5.14,5.14,5.17,5.18,5.17,7838703616
hslinkern_calloc:gol244:1,1217987,228,81,1208105,5.68,5.73,5.69,5.68,5.67,8804442112
hslinkern_calloc:gol251:1,1325625,235,84,1315167,6.34,6.35,6.4,6.26,6.28,9857212416
hslinkern_calloc:gol257:1,1422784,240,86,1411819,7.07,6.83,6.85,6.81,6.85,10783105024
hslinkern_calloc:gol263:1,1524578,246,88,1513094,7.38,7.44,7.43,7.38,7.46,11846361088
hslinkern_calloc:gol269:1,1631115,252,90,1619100,8.03,8.11,8.02,7.99,7.99,12907520000
hslinkern_calloc:gol274:1,1723597,257,91,1711130,8.55,8.71,8.52,8.57,8.53,13942464512
hslinkern_calloc:gol279:1,1819510,262,93,1806583,9.07,9.05,9.2,9.04,9.02,14950146048
hslinkern_calloc:gol284:1,1918917,267,95,1905522,9.64,9.65,9.65,9.65,9.63,16011304960
hslinkern_calloc:gol289:1,2021880,272,96,2008008,10.25,10.21,10.17,10.19,10.2,17129086976
hslinkern_calloc:gol331:1,3035725,312,110,3017520,16.27,16.46,16.26,16.29,16.25,29353385984
hslinkern_calloc:gol365:1,4068883,345,122,4046740,23.21,22.97,22.89,22.9,22.89,33466949632
hslinkern_calloc:gol393:1,5077478,373,131,5051802,29.43,29.61,29.43,29.64,29.77,33466949632
};
    \end{axis}
  \end{tikzpicture}
  \caption{Performance of \autoref{alg:eff-sunflo} on conflict
    hypergraphs of the \textsc{Golomb Subruler} problem with at most
    $2\cdot 10^6$~hyperedges.}
  \label{fig:golomb-time-small}
\end{figure}

\autoref{fig:golomb-time-small} shows the performance of our
kernelization algorithm on conflict hypergraphs of the \textsc{Golomb
  Subruler} problem of size up to $2\cdot 10^6$~hyperedges.  On larger
instances, the implementations based on malloc tries, calloc tries,
and hash tables hit the 32\,GB memory limit of the \texttt{valgrind}
memory measuring tool.  One can observe that the implementation using
the malloc trie is the slowest. This is due to the complicated
initialization procedure required by \autoref{lem:stringtrie}.  The
fastest implementation is the variant using the calloc trie, which is
the same as the malloc trie implementation except that we skip the
intricate initialization of the trie using
\autoref{lem:stringtrie}. Unsurprisingly, the memory usage of the
balanced tree implementation is the lowest, as it grows linearly with
the number of stored elements. Surprisingly, the hash table
implementation of the GNU C++ compiler consumes even more memory than
our calloc trie.

Since the malloc trie, calloc trie, and hash table reach the 32\,GB
memory limit of the \texttt{valgrind} memory measurement tool
between~$2\cdot 10^6$ and $5\cdot 10^6$~input hyperedges, we made
ongoing experiments only with the balanced tree implementation. Thus,
unfortunately, we were unable to see how the running time of our
fastest implementation---using calloc tries---scales to larger
instances.
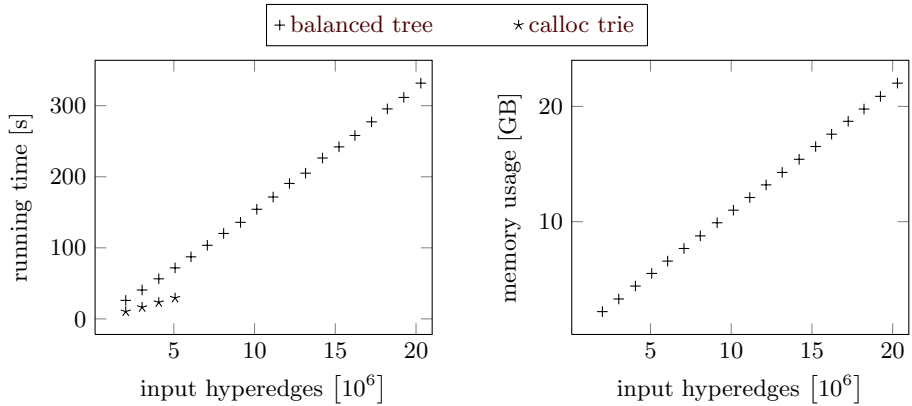
\begin{figure}\small\centering
  \ref{legendsba}

  \begin{tikzpicture}
    \begin{axis}[axis base prefix={axis x base -6 prefix {}}, y
      unit=s, x unit=10^6, xlabel=input hyperedges, ylabel=running
      time, width=0.5\textwidth, change x base, axis base prefix={axis
        x base -6 prefix {}}, xmax=21000000, legend to name=legendsba,
      legend columns=-1, legend style={/tikz/every even column/.append
        style={column sep=1cm}}]

      \addplot[color=black,mark=+,only marks] table [x=In, y=T2, col
      sep=comma,skip coords between index={0}{19}]
      {FILE,In,UBo,LBo,Out,T1,T2,T3,T4,T5,Mem
hslinkern_balanced:gol106:1,100647,95,35,98792,0.88,0.87,0.88,0.88,0.9,123285504
hslinkern_balanced:gol134:1,202742,122,45,199772,1.98,1.98,1.99,1.99,2,233385984
hslinkern_balanced:gol153:1,301378,140,51,297502,3.12,3.15,3.13,3.12,3.13,340340736
hslinkern_balanced:gol169:1,405790,156,56,401058,4.4,4.38,4.41,4.38,4.39,452538368
hslinkern_balanced:gol182:1,506506,168,61,501016,5.61,5.61,5.63,5.59,5.61,561590272
hslinkern_balanced:gol193:1,603728,179,64,597552,6.87,6.87,6.87,6.83,6.84,666447872
hslinkern_balanced:gol204:1,712657,189,68,705755,8.17,8.24,8.24,8.23,8.21,784936960
hslinkern_balanced:gol213:1,810953,198,71,803427,9.53,9.46,9.51,9.59,9.61,890843136
hslinkern_balanced:gol222:1,917896,207,74,909719,10.85,10.97,10.89,10.91,10.82,1007235072
hslinkern_balanced:gol229:1,1007285,213,76,998583,12.16,12.05,12.04,12.04,12.02,1103704064
hslinkern_balanced:gol237:1,1116339,221,79,1107017,13.44,13.49,13.47,13.42,13.47,1222193152
hslinkern_balanced:gol244:1,1217987,228,81,1208105,14.94,15,14.98,15.06,15.03,1331245056
hslinkern_balanced:gol251:1,1325625,235,84,1315167,16.34,16.36,16.3,16.3,16.28,1448685568
hslinkern_balanced:gol257:1,1422784,240,86,1411819,17.75,17.74,17.83,17.65,17.81,1553543168
hslinkern_balanced:gol263:1,1524578,246,88,1513094,19.17,19.04,19.21,19.05,19.24,1663643648
hslinkern_balanced:gol269:1,1631115,252,90,1619100,20.55,20.64,20.72,20.45,20.68,1778987008
hslinkern_balanced:gol274:1,1723597,257,91,1711130,21.74,21.86,21.89,22.04,21.91,1878601728
hslinkern_balanced:gol279:1,1819510,262,93,1806583,23.11,23.2,23.21,23.14,23.15,1983459328
hslinkern_balanced:gol284:1,1918917,267,95,1905522,24.5,24.55,24.51,24.62,24.52,2090414080
hslinkern_balanced:gol289:1,2021880,272,96,2008008,26.15,26.04,26.11,26.19,26.04,2202611712
hslinkern_balanced:gol331:1,3035725,312,110,3017520,41.11,40.66,40.65,40.73,40.56,3300470784
hslinkern_balanced:gol365:1,4068883,345,122,4046740,56.37,56.4,56.2,56.27,56.11,4420349952
hslinkern_balanced:gol393:1,5077478,373,131,5051802,71.74,71.81,71.71,71.9,71.96,5512966144
hslinkern_balanced:gol417:1,6064344,396,139,6035432,87.34,87.33,87.27,86.83,87.25,6583562240
hslinkern_balanced:gol439:1,7074430,418,146,7042383,103.11,103.55,103.25,103.61,103.09,7678275584
hslinkern_balanced:gol459:1,8084845,437,153,8049808,120.56,120.24,120.73,119.72,119.98,8772988928
hslinkern_balanced:gol478:1,9129800,456,159,9091799,136.12,135.91,136.01,136.79,136.22,9906499584
hslinkern_balanced:gol495:1,10137868,472,165,10097113,154.21,154.27,153.41,153.35,153.07,10999115776
hslinkern_balanced:gol511:1,11152000,488,170,11108565,177.1,171.53,177.27,169.65,176.72,12099072000
hslinkern_balanced:gol526:1,12162172,503,175,12116147,186.9,190.63,187.16,191.04,187.27,13193785344
hslinkern_balanced:gol540:1,13158405,516,180,13109895,205.3,205.06,205.03,205.48,204.99,14274867200
hslinkern_balanced:gol554:1,14207607,530,185,14156547,223.16,226.4,223.38,225.03,222.48,15411523584
hslinkern_balanced:gol567:1,15230494,543,189,15177007,240.5,242.09,241.02,241.02,239.52,16521965568
hslinkern_balanced:gol579:1,16217235,554,193,16161458,256.97,258.12,257.45,257.75,257.35,17591513088
hslinkern_balanced:gol591:1,17245700,566,197,17187585,276.13,277.29,276.94,275.85,277.45,18707197952
hslinkern_balanced:gol602:1,18225851,577,201,18165551,293.66,295.5,293.94,294.82,293.32,19770454016
hslinkern_balanced:gol613:1,19242453,588,204,19179927,311.34,311.62,311.05,312.02,311.28,20872507392
hslinkern_balanced:gol624:1,20296172,599,208,20231380,330.55,331.72,331.73,330.45,331.39,22015455232
}; \addlegendentry{balanced tree};

      \addplot[color=black,mark=star, only marks] table [x=In, y=T1,
      col sep=comma,skip coords between index={0}{19}]
      {FILE,In,UBo,LBo,Out,T1,T2,T3,T4,T5,Mem
hslinkern_calloc:gol106:1,100647,95,35,98792,0.33,0.34,0.32,0.34,0.34,355016704
hslinkern_calloc:gol134:1,202742,122,45,199772,0.73,0.74,0.74,0.75,0.73,858333184
hslinkern_calloc:gol153:1,301378,140,51,297502,1.15,1.16,1.16,1.16,1.15,1418272768
hslinkern_calloc:gol169:1,405790,156,56,401058,1.62,1.63,1.62,1.62,1.62,2086215680
hslinkern_calloc:gol182:1,506506,168,61,501016,2.09,2.11,2.09,2.09,2.1,2802393088
hslinkern_calloc:gol193:1,603728,179,64,597552,2.54,2.54,2.57,2.54,2.56,3501793280
hslinkern_calloc:gol204:1,712657,189,68,705755,3.05,3.09,3.1,3.07,3.08,4359528448
hslinkern_calloc:gol213:1,810953,198,71,803427,3.57,3.59,3.59,3.57,3.56,5153300480
hslinkern_calloc:gol222:1,917896,207,74,909719,4.13,4.17,4.12,4.12,4.13,6094921728
hslinkern_calloc:gol229:1,1007285,213,76,998583,4.63,4.59,4.59,4.59,4.59,6848847872
hslinkern_calloc:gol237:1,1116339,221,79,1107017,5.14,5.14,5.17,5.18,5.17,7838703616
hslinkern_calloc:gol244:1,1217987,228,81,1208105,5.68,5.73,5.69,5.68,5.67,8804442112
hslinkern_calloc:gol251:1,1325625,235,84,1315167,6.34,6.35,6.4,6.26,6.28,9857212416
hslinkern_calloc:gol257:1,1422784,240,86,1411819,7.07,6.83,6.85,6.81,6.85,10783105024
hslinkern_calloc:gol263:1,1524578,246,88,1513094,7.38,7.44,7.43,7.38,7.46,11846361088
hslinkern_calloc:gol269:1,1631115,252,90,1619100,8.03,8.11,8.02,7.99,7.99,12907520000
hslinkern_calloc:gol274:1,1723597,257,91,1711130,8.55,8.71,8.52,8.57,8.53,13942464512
hslinkern_calloc:gol279:1,1819510,262,93,1806583,9.07,9.05,9.2,9.04,9.02,14950146048
hslinkern_calloc:gol284:1,1918917,267,95,1905522,9.64,9.65,9.65,9.65,9.63,16011304960
hslinkern_calloc:gol289:1,2021880,272,96,2008008,10.25,10.21,10.17,10.19,10.2,17129086976
hslinkern_calloc:gol331:1,3035725,312,110,3017520,16.27,16.46,16.26,16.29,16.25,29353385984
hslinkern_calloc:gol365:1,4068883,345,122,4046740,23.21,22.97,22.89,22.9,22.89,33466949632
hslinkern_calloc:gol393:1,5077478,373,131,5051802,29.43,29.61,29.43,29.64,29.77,33466949632
}; \addlegendentry{calloc trie};
        
    \end{axis}
  \end{tikzpicture}\hfill{}
  \begin{tikzpicture}
    \begin{axis}[change y base, change x base, axis base prefix={axis
        x base -6 prefix {}}, use units, y SI prefix=giga, y unit=B, x
      unit=10^6, xlabel=input hyperedges, ylabel=memory usage,
      width=0.5\textwidth, xmax=21000000]

      \addplot[color=black,mark=+, only marks] table [x=In, y=Mem, col
      sep=comma,skip coords between index={0}{19}]
      {FILE,In,UBo,LBo,Out,T1,T2,T3,T4,T5,Mem
hslinkern_balanced:gol106:1,100647,95,35,98792,0.88,0.87,0.88,0.88,0.9,123285504
hslinkern_balanced:gol134:1,202742,122,45,199772,1.98,1.98,1.99,1.99,2,233385984
hslinkern_balanced:gol153:1,301378,140,51,297502,3.12,3.15,3.13,3.12,3.13,340340736
hslinkern_balanced:gol169:1,405790,156,56,401058,4.4,4.38,4.41,4.38,4.39,452538368
hslinkern_balanced:gol182:1,506506,168,61,501016,5.61,5.61,5.63,5.59,5.61,561590272
hslinkern_balanced:gol193:1,603728,179,64,597552,6.87,6.87,6.87,6.83,6.84,666447872
hslinkern_balanced:gol204:1,712657,189,68,705755,8.17,8.24,8.24,8.23,8.21,784936960
hslinkern_balanced:gol213:1,810953,198,71,803427,9.53,9.46,9.51,9.59,9.61,890843136
hslinkern_balanced:gol222:1,917896,207,74,909719,10.85,10.97,10.89,10.91,10.82,1007235072
hslinkern_balanced:gol229:1,1007285,213,76,998583,12.16,12.05,12.04,12.04,12.02,1103704064
hslinkern_balanced:gol237:1,1116339,221,79,1107017,13.44,13.49,13.47,13.42,13.47,1222193152
hslinkern_balanced:gol244:1,1217987,228,81,1208105,14.94,15,14.98,15.06,15.03,1331245056
hslinkern_balanced:gol251:1,1325625,235,84,1315167,16.34,16.36,16.3,16.3,16.28,1448685568
hslinkern_balanced:gol257:1,1422784,240,86,1411819,17.75,17.74,17.83,17.65,17.81,1553543168
hslinkern_balanced:gol263:1,1524578,246,88,1513094,19.17,19.04,19.21,19.05,19.24,1663643648
hslinkern_balanced:gol269:1,1631115,252,90,1619100,20.55,20.64,20.72,20.45,20.68,1778987008
hslinkern_balanced:gol274:1,1723597,257,91,1711130,21.74,21.86,21.89,22.04,21.91,1878601728
hslinkern_balanced:gol279:1,1819510,262,93,1806583,23.11,23.2,23.21,23.14,23.15,1983459328
hslinkern_balanced:gol284:1,1918917,267,95,1905522,24.5,24.55,24.51,24.62,24.52,2090414080
hslinkern_balanced:gol289:1,2021880,272,96,2008008,26.15,26.04,26.11,26.19,26.04,2202611712
hslinkern_balanced:gol331:1,3035725,312,110,3017520,41.11,40.66,40.65,40.73,40.56,3300470784
hslinkern_balanced:gol365:1,4068883,345,122,4046740,56.37,56.4,56.2,56.27,56.11,4420349952
hslinkern_balanced:gol393:1,5077478,373,131,5051802,71.74,71.81,71.71,71.9,71.96,5512966144
hslinkern_balanced:gol417:1,6064344,396,139,6035432,87.34,87.33,87.27,86.83,87.25,6583562240
hslinkern_balanced:gol439:1,7074430,418,146,7042383,103.11,103.55,103.25,103.61,103.09,7678275584
hslinkern_balanced:gol459:1,8084845,437,153,8049808,120.56,120.24,120.73,119.72,119.98,8772988928
hslinkern_balanced:gol478:1,9129800,456,159,9091799,136.12,135.91,136.01,136.79,136.22,9906499584
hslinkern_balanced:gol495:1,10137868,472,165,10097113,154.21,154.27,153.41,153.35,153.07,10999115776
hslinkern_balanced:gol511:1,11152000,488,170,11108565,177.1,171.53,177.27,169.65,176.72,12099072000
hslinkern_balanced:gol526:1,12162172,503,175,12116147,186.9,190.63,187.16,191.04,187.27,13193785344
hslinkern_balanced:gol540:1,13158405,516,180,13109895,205.3,205.06,205.03,205.48,204.99,14274867200
hslinkern_balanced:gol554:1,14207607,530,185,14156547,223.16,226.4,223.38,225.03,222.48,15411523584
hslinkern_balanced:gol567:1,15230494,543,189,15177007,240.5,242.09,241.02,241.02,239.52,16521965568
hslinkern_balanced:gol579:1,16217235,554,193,16161458,256.97,258.12,257.45,257.75,257.35,17591513088
hslinkern_balanced:gol591:1,17245700,566,197,17187585,276.13,277.29,276.94,275.85,277.45,18707197952
hslinkern_balanced:gol602:1,18225851,577,201,18165551,293.66,295.5,293.94,294.82,293.32,19770454016
hslinkern_balanced:gol613:1,19242453,588,204,19179927,311.34,311.62,311.05,312.02,311.28,20872507392
hslinkern_balanced:gol624:1,20296172,599,208,20231380,330.55,331.72,331.73,330.45,331.39,22015455232
};

    \end{axis}
  \end{tikzpicture}

  \caption{Performance of \autoref{alg:eff-sunflo} on conflict
    hypergraphs of the \textsc{Golomb Subruler} problem with at most
    $20\cdot 10^6$~hyperedges.}
  \label{fig:golomb-time}
\end{figure}
\autoref{fig:golomb-time} shows that the implementation using the
balanced tree solves \hs 4 instances on conflict hypergraphs of
\textsc{Golomb Subruler} with $20\cdot 10^6$~hyperedges in less than
five minutes and does not even hit the 32GB memory limit of the
\texttt{valgrind} memory measurement tool.

\begin{figure}\small
  \begin{tikzpicture}
    \begin{axis}[change y base, change x base, axis base prefix={axis
        x base -6 prefix {}}, axis base prefix={axis y base -3 prefix
        {}}, x unit=10^6, y unit=10^3, xlabel=input hyperedges,
      ylabel=removed hyperedges, width=0.5\textwidth, legend cell
      align=left, legend pos=north west, xmax=21000000]

      \addplot[color=black,mark=+,only marks] table [x=In, y
      expr=(\thisrow{In}-\thisrow{Out}), col sep=comma, skip coords
      between index={0}{19} ] {FILE,In,UBo,LBo,Out,T1,T2,T3,T4,T5,Mem
hslinkern_balanced:gol106:1,100647,95,35,98792,0.88,0.87,0.88,0.88,0.9,123285504
hslinkern_balanced:gol134:1,202742,122,45,199772,1.98,1.98,1.99,1.99,2,233385984
hslinkern_balanced:gol153:1,301378,140,51,297502,3.12,3.15,3.13,3.12,3.13,340340736
hslinkern_balanced:gol169:1,405790,156,56,401058,4.4,4.38,4.41,4.38,4.39,452538368
hslinkern_balanced:gol182:1,506506,168,61,501016,5.61,5.61,5.63,5.59,5.61,561590272
hslinkern_balanced:gol193:1,603728,179,64,597552,6.87,6.87,6.87,6.83,6.84,666447872
hslinkern_balanced:gol204:1,712657,189,68,705755,8.17,8.24,8.24,8.23,8.21,784936960
hslinkern_balanced:gol213:1,810953,198,71,803427,9.53,9.46,9.51,9.59,9.61,890843136
hslinkern_balanced:gol222:1,917896,207,74,909719,10.85,10.97,10.89,10.91,10.82,1007235072
hslinkern_balanced:gol229:1,1007285,213,76,998583,12.16,12.05,12.04,12.04,12.02,1103704064
hslinkern_balanced:gol237:1,1116339,221,79,1107017,13.44,13.49,13.47,13.42,13.47,1222193152
hslinkern_balanced:gol244:1,1217987,228,81,1208105,14.94,15,14.98,15.06,15.03,1331245056
hslinkern_balanced:gol251:1,1325625,235,84,1315167,16.34,16.36,16.3,16.3,16.28,1448685568
hslinkern_balanced:gol257:1,1422784,240,86,1411819,17.75,17.74,17.83,17.65,17.81,1553543168
hslinkern_balanced:gol263:1,1524578,246,88,1513094,19.17,19.04,19.21,19.05,19.24,1663643648
hslinkern_balanced:gol269:1,1631115,252,90,1619100,20.55,20.64,20.72,20.45,20.68,1778987008
hslinkern_balanced:gol274:1,1723597,257,91,1711130,21.74,21.86,21.89,22.04,21.91,1878601728
hslinkern_balanced:gol279:1,1819510,262,93,1806583,23.11,23.2,23.21,23.14,23.15,1983459328
hslinkern_balanced:gol284:1,1918917,267,95,1905522,24.5,24.55,24.51,24.62,24.52,2090414080
hslinkern_balanced:gol289:1,2021880,272,96,2008008,26.15,26.04,26.11,26.19,26.04,2202611712
hslinkern_balanced:gol331:1,3035725,312,110,3017520,41.11,40.66,40.65,40.73,40.56,3300470784
hslinkern_balanced:gol365:1,4068883,345,122,4046740,56.37,56.4,56.2,56.27,56.11,4420349952
hslinkern_balanced:gol393:1,5077478,373,131,5051802,71.74,71.81,71.71,71.9,71.96,5512966144
hslinkern_balanced:gol417:1,6064344,396,139,6035432,87.34,87.33,87.27,86.83,87.25,6583562240
hslinkern_balanced:gol439:1,7074430,418,146,7042383,103.11,103.55,103.25,103.61,103.09,7678275584
hslinkern_balanced:gol459:1,8084845,437,153,8049808,120.56,120.24,120.73,119.72,119.98,8772988928
hslinkern_balanced:gol478:1,9129800,456,159,9091799,136.12,135.91,136.01,136.79,136.22,9906499584
hslinkern_balanced:gol495:1,10137868,472,165,10097113,154.21,154.27,153.41,153.35,153.07,10999115776
hslinkern_balanced:gol511:1,11152000,488,170,11108565,177.1,171.53,177.27,169.65,176.72,12099072000
hslinkern_balanced:gol526:1,12162172,503,175,12116147,186.9,190.63,187.16,191.04,187.27,13193785344
hslinkern_balanced:gol540:1,13158405,516,180,13109895,205.3,205.06,205.03,205.48,204.99,14274867200
hslinkern_balanced:gol554:1,14207607,530,185,14156547,223.16,226.4,223.38,225.03,222.48,15411523584
hslinkern_balanced:gol567:1,15230494,543,189,15177007,240.5,242.09,241.02,241.02,239.52,16521965568
hslinkern_balanced:gol579:1,16217235,554,193,16161458,256.97,258.12,257.45,257.75,257.35,17591513088
hslinkern_balanced:gol591:1,17245700,566,197,17187585,276.13,277.29,276.94,275.85,277.45,18707197952
hslinkern_balanced:gol602:1,18225851,577,201,18165551,293.66,295.5,293.94,294.82,293.32,19770454016
hslinkern_balanced:gol613:1,19242453,588,204,19179927,311.34,311.62,311.05,312.02,311.28,20872507392
hslinkern_balanced:gol624:1,20296172,599,208,20231380,330.55,331.72,331.73,330.45,331.39,22015455232
};
    \end{axis}
  \end{tikzpicture}\hfill{}
  \begin{tikzpicture}
    \begin{axis}[xlabel={$k:=n-\sqrt n$}, ylabel={output hyperedges},
      y unit=10^6, change y base, axis base prefix={axis y base -6
        prefix {}}, legend cell align=left, legend pos=north west,
      width=0.5\textwidth, xmin=250,xmax=620]

      \addplot[color=black,mark=+,only marks] table [x=UBo, y=Out, col
      sep=comma, skip coords between index={0}{19}]
      {FILE,In,UBo,LBo,Out,T1,T2,T3,T4,T5,Mem
hslinkern_balanced:gol106:1,100647,95,35,98792,0.88,0.87,0.88,0.88,0.9,123285504
hslinkern_balanced:gol134:1,202742,122,45,199772,1.98,1.98,1.99,1.99,2,233385984
hslinkern_balanced:gol153:1,301378,140,51,297502,3.12,3.15,3.13,3.12,3.13,340340736
hslinkern_balanced:gol169:1,405790,156,56,401058,4.4,4.38,4.41,4.38,4.39,452538368
hslinkern_balanced:gol182:1,506506,168,61,501016,5.61,5.61,5.63,5.59,5.61,561590272
hslinkern_balanced:gol193:1,603728,179,64,597552,6.87,6.87,6.87,6.83,6.84,666447872
hslinkern_balanced:gol204:1,712657,189,68,705755,8.17,8.24,8.24,8.23,8.21,784936960
hslinkern_balanced:gol213:1,810953,198,71,803427,9.53,9.46,9.51,9.59,9.61,890843136
hslinkern_balanced:gol222:1,917896,207,74,909719,10.85,10.97,10.89,10.91,10.82,1007235072
hslinkern_balanced:gol229:1,1007285,213,76,998583,12.16,12.05,12.04,12.04,12.02,1103704064
hslinkern_balanced:gol237:1,1116339,221,79,1107017,13.44,13.49,13.47,13.42,13.47,1222193152
hslinkern_balanced:gol244:1,1217987,228,81,1208105,14.94,15,14.98,15.06,15.03,1331245056
hslinkern_balanced:gol251:1,1325625,235,84,1315167,16.34,16.36,16.3,16.3,16.28,1448685568
hslinkern_balanced:gol257:1,1422784,240,86,1411819,17.75,17.74,17.83,17.65,17.81,1553543168
hslinkern_balanced:gol263:1,1524578,246,88,1513094,19.17,19.04,19.21,19.05,19.24,1663643648
hslinkern_balanced:gol269:1,1631115,252,90,1619100,20.55,20.64,20.72,20.45,20.68,1778987008
hslinkern_balanced:gol274:1,1723597,257,91,1711130,21.74,21.86,21.89,22.04,21.91,1878601728
hslinkern_balanced:gol279:1,1819510,262,93,1806583,23.11,23.2,23.21,23.14,23.15,1983459328
hslinkern_balanced:gol284:1,1918917,267,95,1905522,24.5,24.55,24.51,24.62,24.52,2090414080
hslinkern_balanced:gol289:1,2021880,272,96,2008008,26.15,26.04,26.11,26.19,26.04,2202611712
hslinkern_balanced:gol331:1,3035725,312,110,3017520,41.11,40.66,40.65,40.73,40.56,3300470784
hslinkern_balanced:gol365:1,4068883,345,122,4046740,56.37,56.4,56.2,56.27,56.11,4420349952
hslinkern_balanced:gol393:1,5077478,373,131,5051802,71.74,71.81,71.71,71.9,71.96,5512966144
hslinkern_balanced:gol417:1,6064344,396,139,6035432,87.34,87.33,87.27,86.83,87.25,6583562240
hslinkern_balanced:gol439:1,7074430,418,146,7042383,103.11,103.55,103.25,103.61,103.09,7678275584
hslinkern_balanced:gol459:1,8084845,437,153,8049808,120.56,120.24,120.73,119.72,119.98,8772988928
hslinkern_balanced:gol478:1,9129800,456,159,9091799,136.12,135.91,136.01,136.79,136.22,9906499584
hslinkern_balanced:gol495:1,10137868,472,165,10097113,154.21,154.27,153.41,153.35,153.07,10999115776
hslinkern_balanced:gol511:1,11152000,488,170,11108565,177.1,171.53,177.27,169.65,176.72,12099072000
hslinkern_balanced:gol526:1,12162172,503,175,12116147,186.9,190.63,187.16,191.04,187.27,13193785344
hslinkern_balanced:gol540:1,13158405,516,180,13109895,205.3,205.06,205.03,205.48,204.99,14274867200
hslinkern_balanced:gol554:1,14207607,530,185,14156547,223.16,226.4,223.38,225.03,222.48,15411523584
hslinkern_balanced:gol567:1,15230494,543,189,15177007,240.5,242.09,241.02,241.02,239.52,16521965568
hslinkern_balanced:gol579:1,16217235,554,193,16161458,256.97,258.12,257.45,257.75,257.35,17591513088
hslinkern_balanced:gol591:1,17245700,566,197,17187585,276.13,277.29,276.94,275.85,277.45,18707197952
hslinkern_balanced:gol602:1,18225851,577,201,18165551,293.66,295.5,293.94,294.82,293.32,19770454016
hslinkern_balanced:gol613:1,19242453,588,204,19179927,311.34,311.62,311.05,312.02,311.28,20872507392
hslinkern_balanced:gol624:1,20296172,599,208,20231380,330.55,331.72,331.73,330.45,331.39,22015455232
}; \addlegendentry{output};

      \addplot[color=black, domain=0:700, samples=20]{x^3*(2/21)};
      \addlegendentry{$\frac{2}{21}k^3$};
    \end{axis}
  \end{tikzpicture}

  \caption{Size of the resulting problem kernel when
    \autoref{alg:eff-sunflo} is applied to conflict hypergraph of the
    \textsc{Golomb Subruler} problem.}
  \label{fig:golomb-removed}
\end{figure}
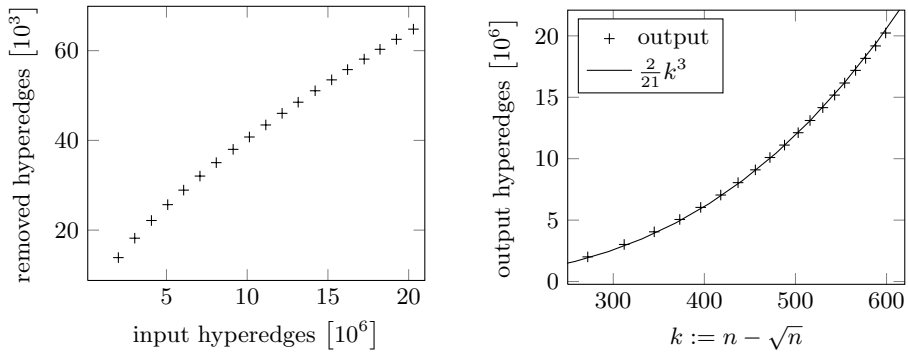

\paragraph{Effect of data reduction.}
As shown in \autoref{fig:golomb-removed}, the kernelization algorithm
removed between $20\cdot 10^3$ and $60\cdot 10^3$~hyperedges from the
input instances.  Thus, although we observed the algorithm to handle
large input instances well, the observed data reduction effect is
rather limited.  This is partly due to the lack of a better upper
bound~$k$ for the size of the sought hitting set: we observed that the
input \hs 4 instances obtained from \textsc{Golomb Subruler} have
roughly $1/12\cdot n^3$~hyperedges. This, for any $n\geq 0$, is
already below the upper bound of~$4!\cdot 4^5\cdot(k+1)^4$ with
$k=\lfloor n-\sqrt n\rfloor$ given by \autoref{thm:optkern-lintime} on
the problem kernel size for \hs4.

The limited effect of the data reduction on \hs4 instances obtained
from \textsc{Golomb Subruler} is also due to the fact that they are
nearly 4-uniform: this prevents hyperedges from getting deleted for
being supersets of smaller hyperedges.  The presence of smaller
hyperedges can significantly influence the output instance size. As an
extreme example, consider a hypergraph containing ${n\choose
  2}$~hyperedges of cardinality two. Then, the output problem kernel
will contain $\bigO(k^2)$~output hyperedges, regardless of how many
more input hyperedges of cardinality~100 there are.  Such phenomena
are not captured in the theoretical upper bound on the problem kernel
size given by \autoref{thm:optkern-lintime}, which is based on the
analysis of uniform hypergraphs.

As shown in \autoref{fig:golomb-removed}, when measuring the size of
the problem kernels in~$k$, we observe that the resulting problem
kernels contain about $2/21\cdot k^3$~hyperedges. Thus, our
empirically measured problem kernel size is lower than the upper bound
of~$3k^3+3k^2$~hyperedges that \citet{SMNW14} have proven using data
reduction rules specifically designed for \textsc{Golomb
  Subruler}. Moreover, our problem kernel is computable in linear
time, while the problem kernel of \citet{SMNW14} takes
$\bigO(k(n+m))$~time. Both problem kernels require the conflict
hypergraph as input.

\paragraph{Summary.} The calloc trie implementation of
\autoref{alg:eff-sunflo} is superior when enough memory is available,
since it is the fastest variant if the C++~environment at hand
implements the allocation of zero-initialized memory using
\texttt{calloc} efficiently. In all other cases, the balanced tree
implementation of \autoref{alg:eff-sunflo} yields a good compromise
between scalability with respect to running time and memory usage.

One can observe that the data reduction effect on nearly uniform
hypergraphs, like those from \textsc{Golomb Subruler}, is rather
limited.  On the other hand, problem kernels for \hs d can be small
even for high values of~$d$ if the input hypergraph is less uniform.

\section{Reducing the number of vertices in
  \boldmath$\bigO(k^{\onepfive d})$ additional time}
\label{sec:fewverts}

This section combines the linear-time computable problem kernel from
\autoref{sec:expl} with techniques of \citet{Abu10} and
\citet[Section~7.3]{Mos10}. This will yield a problem kernel for \hs d
with $\bigO(k^d)$~hyperedges and $\bigO(k^{d-1})$~vertices in
$\bigO(\vernum{}+\edgenum{}+k^{\onepfive d})$~time. Towards this
problem kernel, \autoref{sec:abumoser} first briefly sketches the
running-time bottleneck of the kernelization idea of \citet{Abu10},
which is also a bottleneck in the algorithm of \citet{Mos10}. Then,
\autoref{sec:abumoserimprove} describes our improvements.

\subsection{The approaches of Abu-Khzam and Moser}\label{sec:abumoser} 
\citet{Abu10} has shown a problem kernel for \hs d that comprises
$\bigO(k^{d-1})$~vertices.  \citet[Section~7.3]{Mos10} built upon the
work of \citet{Abu10} to show a problem kernel for \hs d that also
comprises $\bigO(k^{d-1})$~vertices but that, in contrast to the
kernelization algorithm of \citet{Abu10}, yields a subgraph of the
input hypergraph.

The approach of \citet{Abu10} and \citet{Mos10} is as follows.  Given a
hypergraph~$\HG$ and a natural number~$k$, \citet{Abu10} first
computes a maximal \emph{weakly related} set~$\wrel$:

\begin{definition}[\citet{Abu10}]
  A set~$\wrel$ of hyperedges is \emph{weakly related} if every pair
  of hyperedges in~$\wrel$ intersects in at most $d-2$ vertices.
\end{definition}

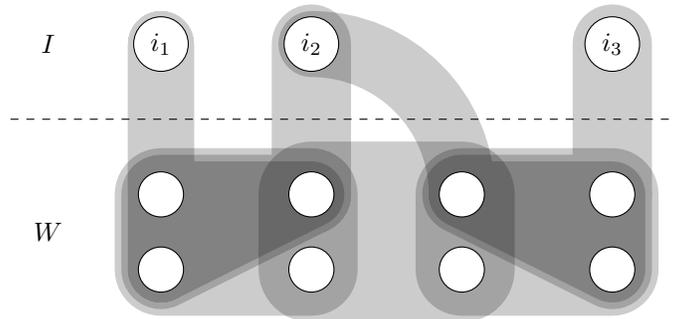
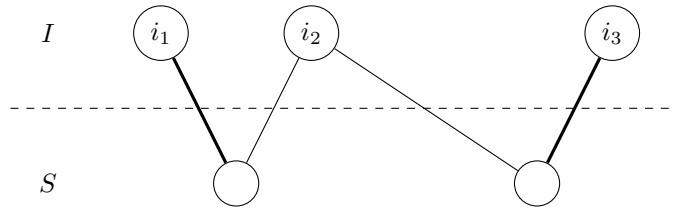
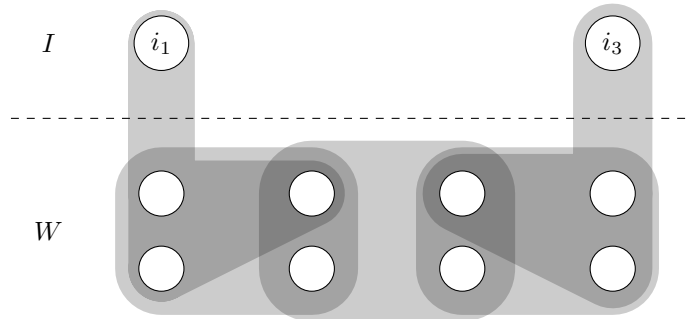
\begin{figure}
  \centering \subfigure[Input hypergraph. The hyperedges fully below
  the dashed line are a maximal weakly related set of hyperedges. The
  vertices above the dashed line are the independent
  set~$I$.]{\hspace{1.5cm}\begin{tikzpicture}[x=2cm]
      \tikzstyle{operator}=[circle,draw=black,fill=white,minimum
      size=6mm]

      \tikzstyle{edge} = [color=black,opacity=.2,line cap=round, line
      join=round, line width=35pt]

      \node[operator] (a1) at (0,0) {}; \node[operator] (a2) at (1,0)
      {}; \node[operator] (a3) at (1,1) {}; \node[operator] (a4) at
      (0,1) {};

      \node[operator] (b1) at (2,1) {}; \node[operator] (b2) at (2,0)
      {};

      \node[operator] (c1) at (3,1) {}; \node[operator] (c2) at (3,0)
      {};

      \node[operator] (i1) at (0,3) {$i_1$}; \node[operator] (i2) at
      (1,3) {$i_2$}; \node[operator] (i4) at (3,3) {$i_3$};

      \begin{pgfonlayer}{background}
        \draw[edge] (a1.center) -- (a2.center) -- (a3.center) --
        (a4.center) -- cycle;

        \draw[edge, line width=40pt] (b1.center) -- (b2.center) --
        (a2.center) -- (a3.center) -- cycle;

        \draw[edge] (c1.center) -- (c2.center) -- (b2.center) --
        (b1.center) -- cycle;

        \draw[edge, line width=25pt] (a1.center) -- (a3.center) --
        (a4.center) -- (i1.center) -- (a1.center);

        \draw[edge, line width=30pt] (i2.center) -- (a3.center) --
        (a4.center) -- (a1.center) -- (a3.center);

        \draw[edge, line width=30pt] (i4.center) -- (c1.center) --
        (b1.center) -- (c2.center) -- (c1.center);

        \draw[edge, line width=25pt] (i2.center) to[out=0,in=90]
        (b1.center) -- (c1.center) -- (c2.center) -- (b1.center);
      \end{pgfonlayer}

      \draw[dashed] (-1,2) -- (3.5,2); \node at (-0.75,3) {$I$}; \node
      at (-0.75,0.5) {$W$};
    \end{tikzpicture}\hspace{1.5cm}}

  \subfigure[The resulting bipartite graph~$B$ with the thick edges
  being a maximum matching.]{\hspace{1.5cm}\begin{tikzpicture}[x=2cm]
      \tikzstyle{operator}=[circle,draw=black,fill=white,minimum
      size=6mm]

      \node[operator] (i1) at (0,3) {$i_1$}; \node[operator] (i2) at
      (1,3) {$i_2$}; \node[operator] (i4) at (3,3) {$i_3$};

      \node[operator] (w1) at (0.5,1) {}; \node[operator] (w2) at
      (2.5,1) {};

      \draw[very thick] (w1) -- (i1); \draw (w1) -- (i2);

      \draw (w2) -- (i2); \draw[very thick] (w2) -- (i4);

      \draw[dashed] (-1,2) -- (3.5,2); \node at (-0.75,3) {$I$}; \node
      at (-0.75,1) {$S$};
    \end{tikzpicture}\hspace{1.5cm}}
  
  \subfigure[The resulting hypergraph with the unmatched vertex~$i_2$
  and its incident hyperedges removed.]{
    \hspace{1.5cm} \begin{tikzpicture}[x=2cm]
      \tikzstyle{operator}=[circle,draw=black,fill=white,minimum
      size=6mm]

      \tikzstyle{edge} = [color=black,opacity=.2,line cap=round, line
      join=round, line width=35pt]

      \node[operator] (a1) at (0,0) {}; \node[operator] (a2) at (1,0)
      {}; \node[operator] (a3) at (1,1) {}; \node[operator] (a4) at
      (0,1) {};

      \node[operator] (b1) at (2,1) {}; \node[operator] (b2) at (2,0)
      {};

      \node[operator] (c1) at (3,1) {}; \node[operator] (c2) at (3,0)
      {};

      \node[operator] (i1) at (0,3) {$i_1$}; \node[operator] (i4) at
      (3,3) {$i_3$};

      \begin{pgfonlayer}{background}
        \draw[edge] (a1.center) -- (a2.center) -- (a3.center) --
        (a4.center) -- cycle;

        \draw[edge, line width=40pt] (b1.center) -- (b2.center) --
        (a2.center) -- (a3.center) -- cycle;

        \draw[edge] (c1.center) -- (c2.center) -- (b2.center) --
        (b1.center) -- cycle;

        \draw[edge, line width=25pt] (a1.center) -- (a3.center) --
        (a4.center) -- (i1.center) -- (a1.center);

        \draw[edge, line width=30pt] (i4.center) -- (c1.center) --
        (b1.center) -- (c2.center) -- (c1.center);
      \end{pgfonlayer}

      \draw[dashed] (-1,2) -- (3.5,2); \node at (-0.75,3) {$I$}; \node
      at (-0.75,0.5) {$W$};
    \end{tikzpicture}\hspace{1.5cm}}
  \caption{Illustration of the kernelization of
    \citet[Lemma~7.16]{Mos10} using \hs 4.}
  \label{fig:moserkern}
\end{figure}

\noindent Whether a given hyperedge~$\Edg$ can be added to a weakly
related set~$\wrel$, \citet{Abu10} checks in $\bigO(d|W|)$~time. After
adding a hyperedge~$\Edg$ to~$W$, he applies data reduction to~$W$ in
$\bigO(2^d|W|\log |W|)$~time that ensures $|W|\leq k^{d-1}$.  Hence,
since $|W|$~never exceeds~$k^{d-1}$, \citet{Abu10} can compute the
maximal weakly related set~$W$ in $\bigO(2^d\cdot k^{d-1}\cdot
(d-1)\log k\cdot\edgenum{})$~time.

Since $|W|\leq k^{d-1}$, it remains to bound the size of the set~$I$
of vertices not contained in hyperedges of~$\wrel$. This is achieved
by the following steps, which are illustrated in
\autoref{fig:moserkern}.  The set~$I$ is an \emph{independent set},
that is, $I$~contains no pair of vertices occurring in the same
hyperedge~\citep{Abu10}. A~bipartite graph~$\BG=(I\uplus S, \Edgs')$
is constructed from the input hypergraph~$\HG=(\Verts,\Edgs)$, where
$S:=\{{e\subseteq V\mid}\exists v\in I\colon\exists w\in \wrel\colon
e\subseteq w, \{v\}\cup e\in\Edgs\}$ and $\Edgs':=\{\{v,e\}\mid v\in
I, e\in S, \{v\}\cup e\in\Edgs\}$. Whereas \citet{Abu10} shrinks the
size of~$I$ using so-called \emph{crown reductions},
\citet[Lemma~7.16]{Mos10} shows that it is sufficient to compute a
maximum matching in~$\BG$ and to remove unmatched vertices in~$I$
together with the hyperedges containing them from the input
hypergraph. The bound of the number of vertices in the problem kernel
is thus $\bigO(k^{d-1})$, since $|\wrel|\leq k^{d-1}$, and, therefore,
$|I|\leq |S|\leq d|W|\leq dk^{d-1}$.

\subsection{Our improvements}\label{sec:abumoserimprove}
We now discuss our running time improvements over the kernelization
algorithms of \citet{Abu10} and \citet{Mos10}.

Given a hypergraph~$\HG$ and a natural number~$k$, we first compute
our problem kernel in $\bigO(\vernum{}+\edgenum{})$~time, leaving
$\bigO(k^d)$~hyperedges in~$\HG$. Afterward, we aim for applying the
ideas of \citet{Abu10} and \citet{Mos10} to reduce the number of vertices
to~$\bigO(k^{d-1})$. However, as discussed in \autoref{sec:abumoser},
the computation of a maximal weakly related set on our reduced
instance already takes $\bigO(2^d\cdot k^{d-1}\cdot(d-1)\log
k\cdot\edgenum{})=\bigO(k^{2d-1}\log k)$~time. We improve the running
time of this step in order to show the following theorem.

\begin{theorem}\label{thm:abufast}
  \hs d has a problem kernel with $d!\cdot d^{d+1} \cdot
  (k+1)^d$~hyperedges and $2\cdot d!\cdot d^{d+1}\cdot
  (k+1)^{d-1}$~vertices computable in $\bigO(d\cdot\vernum{}+2^dd\cdot
  \edgenum{}+(d!\cdot d^{d+1} \cdot (k+1)^{d})^{1.5})$~time.
\end{theorem}

\noindent Note that the problem kernel resulting from
\autoref{thm:abufast} will no longer be expressive in the sense of
\autoref{sec:expressiveness}.  For example, not every minimal hitting
set of size at most~$k$ for the input hypergraph will be a minimal
hitting set for the problem kernel.  This can be observed in
\autoref{fig:moserkern}, where the vertex~$i_2$ might be contained in
a minimal hitting set of the input hypergraph and is absent in the
output hypergraph.

To prove \autoref{thm:abufast}, we compute a maximal weakly related
set~$\wrel$ in linear time and show that our problem kernel already
ensures $|\wrel|\in \bigO(k^{d-1})$ and thus, that further data
reduction on~$W$ is unnecessary. %
To compute a maximal weakly related set in linear time, we employ
\autoref{alg:weakrel}.

After some initialization work in
lines~\ref{lin:init2start}--\ref{lin:init2end}, \autoref{alg:weakrel}
in lines~\ref{lin:foredge}--\ref{lin:elttrue} adds a hyperedge~$\Edg$
to the weakly related set~$\wrel$ if none of the
subsets~$\Core\subseteq\Edg$ with~$|\Core|=d-1$ is a subset of a set
previously added to~$\wrel$.  The information whether~$\Core$ is some
subset of a hyperedge previously added to~$\wrel$ is saved
in~$\intersect[\Core]$.  Note that, in \autoref{lin:elttrue},
\autoref{alg:weakrel} also sets ``$\intersect[\Edg\setminus\Core]\gets
{}$true'' and thus saves which vertices are parts of hyperedges added
to~$\wrel$. We will use this later to quickly reduce the number of
vertices not contained in hyperedges in~$\wrel$.

\begin{algorithm}[t]\small
  \caption{Computation of a maximal weakly related set}
  \label{alg:weakrel}
  \KwIn{Hypergraph~$\HG=(\Verts,\Edgs)$, natural number~$k$.}
  \KwOut{Maximal weakly related set~$\wrel$.}
  $\wrel\gets\emptyset$\nllabel{lin:init2start}\;
  \ForEach(\tcp*[f]{\textrm{Initialization for each
      hyperedge}}){$\Edg\in\Edgs$}{
    \ForEach{$\Core\subseteq\Edg,|\Core|=d-1$}{
      $\intersect[\Core]\gets \bfalse$\tcp*{\textrm{No hyperedges
          in~$\wrel$ contain~$\Core$ yet.}}
      $\intersect[\Edg\setminus\Core]\gets\bfalse$\tcp*{\textrm{The
          vertex in $\Edg\setminus\Core$ is not in~$W$ yet.}}  } }
  \nllabel{lin:init2end}
  
  \ForEach(\nllabel{lin:foredge}){$\Edg\in\Edgs$}{
    \If(\nllabel{lin:wsuit}){$\forall \Core\subseteq e,|\Core|=
      d-1\colon\intersect[\Core]=\bfalse$}{
      $\wrel\gets\wrel\cup\{e\}$\nllabel{lin:addw}\;
      \ForEach{$\Core\subseteq\Edg,|\Core|= d-1$}{
        $\intersect[\Core]\gets \btrue$\nllabel{lin:intertrue}\;
        $\intersect[\Edg\setminus\Core]\gets\btrue$\nllabel{lin:elttrue}\;
      } } } \Return{\wrel}\nllabel{lin:retw}\;
\end{algorithm}

\begin{lemma}\label{lem:weakrel}
  Given a hypergraph $\HG$, a maximal weakly related set is computable
  in $\bigO(d\cdot\vernum{}+d^2\cdot \edgenum{})$~time.
\end{lemma}

\begin{proof}
  First, observe that the set~$W$ returned in \autoref{lin:retw} of
  \autoref{alg:weakrel} when applied to~$\HG=(\Verts,\Edgs)$ is indeed
  weakly related: let~$w_1\ne w_2\in \Edgs$ intersect in more than
  $d-2$~vertices and assume that $w_1$~is added to~$W$ in
  \autoref{lin:addw}. Let $\Core:=w_1\cap w_2$. Obviously,
  $|\Core|=d-1$. Hence, when $w_1$ is added to~$W$, we apply
  ``$\intersect[\Core]\gets\btrue$'' in
  \autoref{lin:intertrue}. Therefore, when $\Edg=w_2$ is considered in
  \autoref{lin:foredge}, the condition in \autoref{lin:wsuit} does not
  hold, which implies that $w_2$~is not added to~$W$ in
  \autoref{lin:addw}.  In the same way it follows that each hyperedge
  is added to~$W$ if it does not intersect any hyperedge of~$W$ in
  more than $d-2$~vertices. Therefore, $W$~is~maximal.

  Now, \autoref{alg:weakrel} works as follows.  We use
  \autoref{lem:stringtrie} to look up values in~$\intersect[]$ in
  $\bigO(d)$~time.  To this end, like in the proof of
  \autoref{lem:kern-schnell}, we represent vertex subsets of size at
  most~$d$ as sorted sequences of length at most~$d$. Thus, we first
  sort each hyperedge of~$\HG$ in $\bigO(\edgenum{}\cdot d\log
  d)$~total time.  To apply \autoref{lem:stringtrie} to create the
  associative array~$\intersect[]$, we need a list~$L$ of all values
  that we are going to store values for. As~$L$, we use the list that,
  for each hyperedge~$\Edg$ of~$\HG$ and each vertex~$v\in e$,
  contains~$\Edg\setminus \{v\}$ and $\{v\}$.  Of course,
  $\Edg\setminus\{v\}$~can be computed in $\bigO(d)$~time from~$\Edg$
  so that~$\Edg\setminus\{v\}$ is sorted. It follows that $L$~contains
  at most $2d\cdot \edgenum{}$~elements and is computable in
  $\bigO(d^2\edgenum{})$~time. Hence, by \autoref{lem:stringtrie}, we
  can build the associative array~$\intersect[]$ in
  $\bigO(d\vernum{}+d^2\cdot\edgenum{})$~time and looking up values in
  $\intersect[]$ works in $\bigO(d)$~time for elements of~$L$.

  Now, the initialization in lines
  \ref{lin:init2start}--\ref{lin:init2end} works in
  $\bigO(d^2\cdot\edgenum{})$~time.  Finally, for every hyperedge, the
  body of the for-loop in \autoref{lin:foredge} can be executed in
  $\bigO(d^2)$~time by doing $\bigO(d)$-time look-ups for each of the
  $2d\cdot\edgenum{}$~sets.
\end{proof}
\noindent We can now prove \autoref{thm:abufast} by showing how to
compute a problem kernel with $\bigO(k^{d-1})$~vertices in
$\bigO(\vernum{}+\edgenum{}+k^{\onepfive d})$~time.

\begin{proof}[Proof of \autoref{thm:abufast}]
  It is shown in \autoref{thm:optkern-lintime} that \hs d has a
  problem kernel with $d!\cdot d^{d+1} \cdot (k+1)^d$~hyperedges that
  is computable in $\bigO(d\vernum{}+2^dd\cdot \edgenum{})$ time.  It
  remains to show that, in additional $\bigO(d!\cdot d^{d+1} \cdot
  (k+1)^{d})^{1.5}$~time, the number of vertices of a hypergraph~$\HG$
  output by \autoref{alg:eff-sunflo} can be reduced to $2\cdot d!\cdot
  d^{d+1}\cdot (k+1)^{d-1}$. To this end, we follow the approaches of
  \citet{Abu10} and \citet{Mos10} as discussed in \autoref{sec:abumoser} and as
  illustrated in \autoref{fig:moserkern}.

  First, we compute a maximal weakly related set~$W$ in~$\HG$ in
  $\bigO(d\cdot \vernum{}+d^2\cdot \edgenum{})$~time using
  \autoref{alg:weakrel}. We show that~$|W|\leq d!\cdot d^d\cdot
  (k+1)^{d-1}$. To this end, consider the
  hypergraph~$\HG_{\dslice}:=(\Verts,W_\dslice)$ for $1\leq\dslice\leq
  d$, where $W_\dslice$ is the set of cardinality-$\dslice$ hyperedges
  in~$W$. Since~$\HG$ has been output by \autoref{alg:eff-sunflo}, we
  know that, by \autoref{lem:few-petals}, $\HG_{\dslice}$ has no
  sunflowers with more than ${d(k+1)}$~petals.  Moreover, since every
  pair of hyperedges in~$W$ intersects in at most $d-2$~vertices, also
  each pair of hyperedges of~$\HG_\dslice$ intersects in at most
  $d-2$~vertices. Hence, by \autoref{lem:largeflowers} with~$\cone=2$
  and~$\ctwo=d(k+1)$, we know that $\HG_\dslice$ for $\dslice\geq2$
  has at most $\dslice!d^{\ell-1}(k+1)^{\ell-1}$ hyperedges. Moreover,
  $\HG_1$ contains at most $d(k+1)$~hyperedges, as they form a
  sunflower with empty core.  Therefore, $|W|\leq d!\cdot d^d\cdot
  (k+1)^{d-1}$.

  Next, we construct a bipartite graph~$\BG=(I\uplus S, \Edgs')$ from
  the input hypergraph~$\HG=(\Verts,\Edgs)$, where
  \begin{enumerate}[i)]
  \item $I$ is the set of vertices in~$\Verts$ not contained in any
    hyperedge in~$W$, which is an independent set~\citep{Abu10},
  \item $S:=\{e\subseteq V\mid\exists v\in I\colon\exists w\in
    \wrel\colon e\subseteq w, \{v\}\cup e\in\Edgs\}$, and
  \item $\Edgs':=\{\{v,e\}\mid v\in I, e\in S, \{v\}\cup e\in\Edgs\}$.
  \end{enumerate}
  This can be done in $\bigO(d^2\cdot\edgenum{})$~time by exploiting
  the information stored in the associative array~$\intersect[]$
  computed by \autoref{alg:weakrel}: for each $\Edg\in\Edgs$ with
  $|\Edg|=d$ and each~$v\in\Edg$, add $\{v,\Edg\setminus\{v\}\}$ to
  the graph~$\BG$ if and only if
  $\intersect[\Edg\setminus\{v\}]=\btrue$ and
  $\intersect[\{v\}]=\bfalse$. In this case, it follows that~$\Edg$
  can be partitioned into
  \begin{enumerate}[i)]
  \item a subset~$\Edg\setminus\{v\}$ of a hyperedge of~$W$, since
    $\intersect[\Edg\setminus\{v\}]=\btrue$, and
  \item the vertex~$v$, which is not contained in any hyperedge
    in~$W$, since we have $\intersect[\{v\}]=\bfalse$, and, hence, is
    contained in~$I$.
  \end{enumerate}
  \looseness=-1 Thus, $\Edg$ clearly satisfies the definition
  of~$\Edgs'$.  Observe that the graph~$\BG$ constructed in this way
  contains at most
  $|\Edgs|=\edgenum{}$~edges. %
  It remains to shrink~$I$ so that it contains at most
  $|S|$~vertices. Then, the number of vertices in the output
  hypergraph will be at most~$d|W|+|I|\leq d|W|+|S|\leq 2d|W|=2\cdot
  d!\cdot d^{d+1}\cdot (k+1)^{d-1}$.  This, as shown by
  \citet[Section~7.3]{Mos10}, is achieved by computing a maximum
  matching in~$\BG$ and deleting from~$\HG$ the unmatched vertices
  in~$I$ and the hyperedges containing them.  To analyze the running
  time of computing the maximum matching, recall that the number of
  edges in~$\BG$ is at most $m\leq d!\cdot d^{d+1} \cdot (k+1)^d$ and
  that the number~$|I|+|S|$ of vertices is at most twice as
  much. Hence, a maximum matching in~$\BG$ can be computed in
  $\bigO(\sqrt{|I\uplus S|}\cdot|\Edgs'|) =\bigO(d!\cdot d^{d+1} \cdot
  (k+1)^{d})^{1.5}$~time using the algorithm of Hopcroft and
  Karp~\citep[Theorem~16.4]{Sch03}.
\end{proof}

\section{Conclusion}
We have given an understanding of expressive kernelization for \hs d
and have shown, as earlier claimed by \citet{NR03}, that a problem
kernel for \hs d with $\bigO(k^d)$~hyperedges and vertices can be
computed in linear time.  Using the linear-time computable problem
kernel for \hs d, we have improved the worst-case running
times %
of the $\bigO(k^{d-1})$-vertex problem kernels by \citet{Abu10} and
\citet{Mos10}.

Our experiments have shown that the kernelization algorithm runs
efficiently, yet the observed data reduction effect on the nearly
uniform hypergraphs occurring in the construction of Golomb rulers was
limited.

An interesting question is whether a problem kernel with
$\bigO(k^{d-1})$~vertices and $\bigO(k^d)$~hyperedges for \hs d can be
computed in linear time. Answering this question would merge the best
known results for problem kernels for \hs d. However, to date, all
$\bigO(k^{d-1})$-vertex problem kernels for \hs d that we are aware
of, that is, the problem kernels by \citet{Abu10} and \citet{Mos10}, involve the
computation of maximum matchings. This seems to be difficult to avoid
this bottleneck.

\newcommand{\noopsort}[1]{}

\end{document}